\documentclass{article}

\usepackage{comment}
\usepackage{amsmath,amssymb,xcolor}
\usepackage{enumitem}
\usepackage{graphicx}
\usepackage{lpic}
\usepackage{float}

\usepackage[backend=biber,style=alphabetic,maxnames=10]{biblatex}
\usepackage{amsthm,thmtools}
\declaretheoremstyle[qed=$\lrcorner$,bodyfont=\it]{it}
\declaretheoremstyle[qed=$\lrcorner$,bodyfont=\rm]{rm}

\declaretheorem[style=it,numberwithin=subsection]{lemma}

\makeatletter
\let\c@equation=\c@lemma
\let\theequation=\thelemma
\makeatother

\declaretheorem[style=it,numberlike=lemma]{corollary}
\declaretheorem[style=it,numberlike=lemma]{theorem}
\declaretheorem[style=it,numberlike=lemma]{proposition}

\declaretheorem[style=it,numbered=no]{remark}
\declaretheorem[style=it,numbered=no]{definition}

\declaretheorem[style=it,name=Theorem]{theoremA}

\usepackage[hyperfootnotes=false,
    colorlinks,
    linkcolor={blue},
    citecolor={blue},
    urlcolor={blue}
    ]{hyperref}

\newcommand{\XYAB}{XY\mkern -4.5mu AB}

\def\shiftup{\vspace{-5ex}\par}
\makeatletter

\usepackage{xcolor}
\def\exclam#1{\ifmmode\relax\else\marginpar{\textcolor{#1}{\fbox{!!!}}}\fi}
\long\def\red#1{{\textcolor{red}{\exclam{red}#1}}}

\let\tmp=\phi \let\phi=\varphi \let\varphi=\tmp
\let\tmp=\epsilon \let\epsilon=\varepsilon \let\varepsilon=\tmp

\usepackage{bm}

\usepackage{accents}

\def\mkmathletter#1#2{%
    \expandafter\gdef\csname#1#2\endcsname%
           {\ensuremath{\csname math#2\endcsname{#1}}}%
}

\edef\letters{a,b,c,d,e,f,g,h,i,j,k,l,m,n,o,p,q,r,s,t,u,v,w,x,y,z}
\edef\Letters{A,B,C,D,E,F,G,H,I,J,K,L,M,N,O,P,Q,R,S,T,U,V,W,X,Y,Z}
\edef\Lletters{\letters,\Letters}

\def\mkmathletters#1#2{%
        \@for\ltr:={#2}\do{\expandafter\mkmathletter\ltr{#1}}%
}

\mkmathletters{bb}{\Letters}    
\mkmathletters{cal}{\Lletters}   
\mkmathletters{bf}{\Lletters}   
\mkmathletters{sf}{\Lletters}   

\renewcommand{\to}[1][]{\stackrel{#1}{\rightarrow}}

\def\indep{\,\makebox[0cm][l]{$\bot$}\,\bot\,}

\def\<{\left\langle}
\def\>{\right\rangle}

\def\vect#1{\left(\!\!\!\begin{array}{c}#1\end{array}\!\!\!\right)}

\def\st{\;{\bm :}\;}
\def\set#1{\left\{#1\right\}}

\let\dmo=\DeclareMathOperator

\dmo\argmin{\arg\min}
\dmo\argmax{\arg\max}
\dmo\aut{\mathrm{Aut}}
\dmo\dist{\mathrm{dist}}
\dmo\emp{\qbf}
\dmo\Hom{\mathrm{Hom}}
\dmo\ing{\mathsf{ing}}
\dmo\ks{\mathsf{ks}}
\let\L=\underfined
\dmo\L{\mathrm{L}}
\dmo\mmrv{\mathcal{E}}
\dmo\prob{\mathrm{Prob}}
\dmo\supp{\mathrm{supp}}
\dmo\tr{\mathrm{tr}}
\dmo\ext{\mathrm{Ext}}
\dmo\extgr{\mathrm{Ext}_{\mathrm{gr}}}
\dmo\extgru{\mathrm{Ext}_{\mathrm{gr}}^{\wedge}}
\dmo\extgrl{\mathrm{Ext}_{\mathrm{gr}}^{\vee}}
\dmo\closure{\mathrm{Closure}}
\dmo\convexhull{\mathrm{ConvexHull}}
\dmo\nbd{\mathrm{Nbd}}
\dmo\gr{\mathrm{Gr}}
\def\1{\bm{1}}
\makeatother

\AtBeginBibliography{\emergencystretch=2em}

\title{Beyond Mutual Information: Extension Profiles and Shape
  Functions of Random Variable Pairs}

\author{Rostislav Matveev and
  Andrei Romashchenko}

\begin{document}
\maketitle

\begin{abstract}
We study the extension profile of a pair of jointly distributed
finite-valued random variables $(X,Y)$, defined as the set of all
triples  of numbers 
\[
\bigl(H(X|W), H(Y|W), I(X:Y|W)\bigr)
\]
obtained by extending the pair with an auxiliary random variable $W$.
This object captures structural properties of a joint distribution that
are not determined by the entropies of $X$ and $Y$ and their mutual information.

We associate with the extension profile its \emph{shape function}
defined in terms of the  Legendre--Fenchel transform of the profile’s lower boundary.
We establish upper and lower bounds on the shape function in terms of classical
information-theoretic quantities.
For random pairs $(X,Y)$ whose joint distribution is uniform on its support,
we relate the shape function to combinatorial and spectral properties of 
the associated bipartite support graph.
We show that for random pairs uniformly distributed on the edges of balanced expander graphs, 
the shape function is within \(O(\log H(XY))\) of the pointwise infimum over all pairs, 
which proves the conjecture stated in \cite{matveev2026spectral}.

These results connect the extension profile with non-Shannon-type
information inequalities, G\'acs--Körner common information, 
and spectral graph theory.
We apply these techniques to study an approximate version
of a question posed by Zhen Zhang (\cite{zhang3new}) concerning
the characterization of pairs with maximally non-extractable
mutual information.
\end{abstract}

\section{Introduction}\label{s:intro}

Shannon mutual information is a classical and well-studied measure of
correlation between two random variables.  In many applications,
however, we may need a deeper understanding of the structure of
dependence and, in particular, of the extent to which mutual
information can be ``materialized'' by an explicit shared component.
The very notion of materializing mutual information is subtle and
admits several inequivalent formalizations.  The first systematic
investigations of this question date back to the 1970s and led to the
concepts of common information, as introduced by G\'acs and Körner
\cite{gacs1973common} and by Wyner \cite{wyner1975common}.  More
refined measures of the materializability of mutual information have
since been proposed; see, for example, Zhang \cite{zhang3new} and the
subsequent development of the notion of the tension region in
\cite{prabhakaran2014assisted,li2017extended,csirmaz2023short}.
These questions are closely connected to the structure of the entropy
region and to non-Shannon-type information inequalities.
In this work, we propose a new geometric perspective on common
information and its generalizations.  Using methods from spectral
graph theory, we construct a family of examples exhibiting extremal,
nearly minimal, materializability of mutual information.

\subsection{Extension profile and shape function of a pair of random
  variables}\label{s:ext-shape}
For a pair of jointly distributed random variables $(X,Y)$ with finite
alphabets one can evaluate the entropies $H(X)$, $H(Y)$ and the mutual
information of the pair $I(X:Y)$.  These three numbers describe the
information-theoretic behavior of the pair $(X,Y)$.  However, there
exist finer invariants of the pair that cannot be expressed solely as
functions of these three quantities.  For example, one may ask what
joint distributions of the triple $(X,Y,W)$ are possible when an
additional random variable $W$ is introduced.  What joint entropy
profiles can arise for such triples?  The auxiliary variable $W$ may
be viewed as a probe that tests the internal structure of the
distribution of $(X,Y)$.

The \emph{extension profile} of the pair, defined as the set
\[
  \ext(X,Y)=\set{\vect{H(X|W)\\H(Y|W)\\I(X:Y|W)}}\subset\Rbb^{3}
\]
for all $W$ jointly distributed with $(X,Y)$, formalizes this idea.
This set is closely related to the tension region defined and studied
in~\cite{prabhakaran2014assisted,li2017extended,csirmaz2023short}.  In
the setting of Kolmogorov complexity, a similar notion of extension
profile was studied in \cite{chernov2002upper}; the corresponding
object was referred to there as the set of \emph{admissible triples}.

It can be shown, see~\cite{csiszar2011information} and
Proposition~\ref{p:ext-compact-convex} below, that for every pair
$(X,Y)$ the extension profile is a compact and convex subset of
$\Rbb^{3}$.  One can establish a lower and an upper bound for the
extension profile
\[
  \ext_{\min}(X,Y)\subset\ext(X,Y)\subset\ext_{\max}(X,Y)
\]
where sets $\ext_{\min}(X,Y),\ext_{\max}(X,Y)\subset\Rbb^{3}$ 
are defined in terms of the entropies of $(X,Y)$. 
The upper bound for $\ext(X,Y)$ follows from the standard Shannon
inequalities (since entropies of every triple $(X,Y,W)$ must respect
these inequalities).  The  lower bound for $\ext(X,Y)$ mentioned above
is obtained via direct constructions that allow for any pair $(X,Y)$
to build a $W$ with a certain entropy profile $e(X,Y|W)$ (e.g., one
can let $W=X$, or $W=XY$, and so on).

The correspondence $(X,Y)\mapsto\ext(X,Y)$ has the property of
superadditivity in the sense that
\[
  \ext(X'',Y'')\supset\ext(X,Y)\oplus\ext(X',Y')
\]
where $X''$ and $Y''$ are joints of independent copies of $X$, $X'$
and $Y$, $Y'$, respectively, and $A\oplus B$ stands for the Minkowski
addition of subsets in $\Rbb^{3}$.  

In general, the inclusion above is proper. However, $\ext$ is additive
``in the most important directions''. To make the latter statement
precise, we introduce the notion of the \emph{shape function} or
simply \emph{shape} of the pair of random variables $(X,Y)$.  It
associates with every pair $(X,Y)$ a function $\Scal(X,Y)$ defined on
$(\alpha,\beta)\in[0,1]^2$ that captures the shape of the lower part
of the body $\ext(X,Y)$.  Technically,
\[
  \Scal(X,Y)(\alpha,\beta)
  =
  \sup\set{\alpha\cdot H(X|W)+\beta\cdot H(Y|W)-I(X:Y|W)}
\]
where the supremum is over all $W$'s jointly distributed with $(X,Y)$.
For the shape function we prove additivity
\[
  \Scal(X'',Y'')=\Scal(X,Y)+\Scal(X',Y')
\]
for independent joints, and the bounds
\[
  \Scal_{\min}(X,Y)\leq\Scal(X,Y)\leq\Scal_{\max}(X,Y)
\]
where the ``extremal'' shape functions $\Scal_{\min}(X,Y)$ and
$\Scal_{\max}(X,Y)$ depend only on $H(X)$, $H(Y)$, and $I(X:Y)$.  We
also show that the shape function satisfies a version of the chain
rule.

Thus, the introduced shape function is an entropy-like invariant of
the pair. Importantly, it captures structural properties of $(X,Y)$
beyond those provided by entropy.  We believe that the shape function
may prove useful in the study of information inequalities, entropy
regions, and the structural analysis of joint distributions.

\subsection{Shape function for uniform pairs of random variables
  supported on graphs.}\label{s:shape-unif}
We say that a pair $(X,Y)$ of random variables is \emph{uniform on its support} if
$X$, $Y$ and their joint $XY$ are each uniform on their respective
supports. In that case, the supports
\begin{equation*}
  \Xsf:=\supp X,\qquad\Ysf:=\supp Y,\qquad \Esf:=\supp XY
\end{equation*}
form a biregular bipartite graph $\Gsf=(\Xsf\sqcup\Ysf,\Esf)$.  Let
$\set{\lambda_{i}\st i=1,2,\dots}$ be the spectrum of $\Gsf$, with
eigenvalues indexed in the non-increasing order and counted with
multiplicity.

For a pair uniform on its support, the first eigenvalue $\lambda_1$
can be found as
\[
\log\lambda_{1}=\frac12\bigl(H(X|Y)+H(Y|X)\bigr).
\] 
On the other hand, the value of $\lambda_{2}$ can not be expressed as
a function of the entropies of $(X,Y)$ and captures more subtle
properties of the pair.

We investigate the relation between the spectrum of $\Gsf$ and the shape
function of the pair supported on $\Gsf$. We show that smaller values
of $\lambda_{2}$ force the shape function to be smaller.

\subsection{Maximality and minimality of the extension profile and the
  shape function}
We use this new tool to investigate for what pairs the extension profile is
minimal or maximal possible.  To this end we consider pairs $(X,Y)$ of
random variables uniform on their support and study the effect
the spectrum of $\Gsf$ has on the size of the extension profile and the shape
function. Smaller values of the second largest eigenvalue of $\Gsf$
force the extension profile and the shape function to be smaller. This
dependence is especially demonstrative in the two extreme cases.

On the one hand we consider pairs supported on balanced expanders,
that is, families of graphs such that
\[
  \lambda_{2}\leq C\cdot\sqrt{\lambda_1}
\]
where $\lambda_1$ and $\lambda_{2}$ are the largest and the second
largest eigenvalues of $\Gsf$. In that case we show that the shape
function attains its lower bound up to an error logarithmic in
$H(XY)$.

On the other hand, we show that for any, not necessarily uniform, pair
$(X,Y)$ the shape function $\Scal(X,Y)$ attains their upper bounds if
and only if the pair $(X,Y)$ satisfies the G\'acs--Körner condition on
extractability of the mutual information\footnote{%
  In terms of \cite{gacs1973common}, this property means that the
  mutual information of $(X,Y)$ is equal to the \emph{common
    information} of this pair.}, i.e., there exists a $W$ such that
\[
H(W|X) = 0,\ H(W|Y) = 0,\ \text{and}\ H(W)=I(X:Y).
\]
Thus we prove the following theorem, see
  Theorem~\ref{p:maxprofile} and
  Corollary~\ref{cor:L<Lmin}.
\begin{theoremA}\label{th:shape-minmax}
  \begin{enumerate}[label=(\roman*)]
  \item For every pair of random variables $(X,Y)$ the mutual information is
    equal to the G\'acs--Körner
    information 
    if and only if
    \[
      \Scal(X,Y)=\Scal_{\max}(X,Y)
    \]
    In that case, if $(X,Y)$ is uniformly supported on an admissible%
    \footnote{Admissible graphs satisfy mild technical conditions: they
      contain at least three edges and the density of edges does not
      exceed $1/2$. Details can be found in Section~\ref{s:graphs}.} %
    graph $\Gsf$, then
    \[
      \lambda_{1}(\Gsf)=\lambda_{2}(\Gsf)
    \]
  \item If $(X,Y)$ is supported on a balanced expander, then
    \[
      \Scal(X,Y)\leq\Scal_{\min}(X,Y)+O\bigl(\log H(XY)\bigr)
    \]
  \end{enumerate}
\end{theoremA}

\subsection{Ingleton inequality and friends}
Since the
discovery of the first non-Shannon-type inequality by Yeung and
Zhang,~\cite{zhang1998characterization}, many other non-Shannon-type
inequalities were discovered, including both sporadic inequalities as
well as infinite families (see
\cite{makarychev2002new,matus2007infinitely,dougherty2011non} for
infinite families of non-Shannon-type inequalities).  Every
non-Shannon-type inequality for four random variables, up to
permutation of the variables, can be interpreted as a perturbation of
the Ingleton inequality
\begin{equation*}
  I(X:Y|A)+I(X:Y|B)+I(A:B)-I(X:Y)\geq0
\end{equation*}
Such a perturbation has the form
\begin{equation*}
  I(X:Y|A)+I(X:Y|B)+I(A:B)-I(X:Y)\geq\Psi(X,Y,A,B)
\end{equation*}
where $\Psi$ is some quantity dependent on the quadruple $(X,Y,A,B)$.

\bigskip

\noindent
We consider two particular perturbations of the Ingleton inequality.
\paragraph{\bf The MMRV inequality:}
In~\cite{makarychev2002new} the authors prove the following inequality
valid for any five jointly distributed random variables:
\begin{equation*}
  \begin{aligned}
    I(X:Y&|A)+I(X:Y|B)+I(A:B)-I(X:Y)\geq\\
         &\geq
           -(I(X:Y|W)+I(X:W|Y)+I(W:Y|X))
  \end{aligned}
\end{equation*}
The random variable $W$ appears only on the right-hand side. Therefore
we can define
\begin{equation}
\label{eq:epsilon(x,y)}
  \mmrv(X,Y)=\inf_{W}\set{I(X:Y|W)+I(X:W|Y)+I(W:Y|X)}
\end{equation}
where the infimum is taken over all $W$'s jointly distributed with
$(X,Y)$, and rewrite the MMRV inequality as
\begin{equation}
  \label{eq:mmrv}
  I(X:Y|A)+I(X:Y|B)+I(A:B)-I(X:Y)\geq - \mmrv(X,Y)
\end{equation}
We refer to the quantity $\mmrv(X,Y)$ as \emph{entanglement}.%
\footnote{We thank Alexander Shen for this very fitting term.}  %
Note that the entanglement $\mmrv(X,Y)$ can not be expressed as a
function of entropies of $X$, $Y$ and their joint, while it represents
more intricate properties of the pair $(X,Y)$.  It can be interpreted
as a measure of the extent to which the mutual information between $X$
and $Y$ can be extracted --- 
the smaller $\mmrv(X,Y)$ is, the more extractable mutual information there is. 
The value $0$ corresponds to the case
where mutual information is completely extractable. By Double
Markovity Lemma,~\cite[Exercise 16.25,
pp.~392--393]{csiszar2011information}, see
also~\cite{makarychev2002new,kaced2013conditional}, if $\mmrv(X,Y)=0$,
then there exists an extension $(X,Y,W)$ such that
\begin{equation*}
  I(X:Y|W)=H(W|X)=H(W|Y)=0
\end{equation*}\label{gacs-korner}
We refer to the latter condition as \emph{the G\'acs--Körner extractable
  mutual information condition}, see~\cite{gacs1973common}.  
In this case the pair $(X,Y)$ has a very simple structure: it decomposes as
\begin{equation*}
  (X,Y)=(X'W,Y'W)
\end{equation*}
where $X'\indep Y' | W$. 
Somewhat loosely one might say
that in this case $X$ and $Y$ are independent except for a shared
common part $W$. This justifies the name \emph{entanglement} for the
quantity $\mmrv(X,Y)$.

\paragraph{The spectral inequality:} In~\cite{matveev2026spectral}
the authors prove the following inequality, which is valid for any uniform
on the support pair $(X,Y)$ of random variables such that
$I(X:Y)\geq\epsilon_{0}>0$ holds for all pairs in the family for some
fixed constant $\epsilon_{0}$ and any, not necessarily uniform,
extension $(X,Y,A,B)$
\begin{equation*}
  \label{eq:mr}
  I(X:Y|A)+I(X:Y|B)+I(A:B)-I(X:Y)\geq\log\frac{\lambda_{1}}{\lambda_{2}^{2}}
  - O(\log H(\XYAB))
\end{equation*}
where $\lambda_{i}$ are eigenvalues of the graph $\Gsf$ supporting the
uniform pair $(X,Y)$ and constants implicit in $O(\log H(\XYAB))$
depend on $\epsilon_{0}$.

Note that the smaller the value of $\lambda_{2}(\Gsf)$ is, the
stronger the spectral inequality becomes.  It is especially
strong for families of pairs $(X,Y)$ uniformly distributed on
\emph{balanced expanders}, as described in Section~\ref{s:expanders}.
In that case the spectral inequality rewrites to
\begin{equation*}
  I(X:Y|A)+I(X:Y|B)+I(A:B)-I(X:Y)\geq-O(\log H(\XYAB))
\end{equation*}

\bigskip

As an application of the shape function we derive a spectral
bound on the entanglement of the pair $(X,Y)$ and show that smaller
values of $\lambda_{2}(\Gsf)$ force inequality~\eqref{eq:mmrv} to
become weaker. More specifically we prove the following
statement, see Theorem~\ref{th:entanglement-spectral} below.
This statement was formulated as Conjecture~B
in~\cite{matveev2026spectral}.
\begin{theoremA}\label{th:main-intro}
  For any pair $(X,Y)$ of random variables that is uniformly
  distributed on an admissible graph $\Gsf$ it holds
  \begin{equation*}
    \mmrv(X,Y)
    \geq
    \min\set{I(X:Y),2\log\frac{\lambda_{1}(\Gsf)}{\lambda_{2}(\Gsf)}}
    -O\bigr(\log H(XY)\bigl)
  \end{equation*}
  where constants in the $O\bigr(\log H(XY)\bigl)$-term are universal.
\end{theoremA}

Thus, for distributions on expanders, inequality~\eqref{eq:mmrv} is
weaker than
\begin{equation*}
  \begin{split}
    I(X:Y|A)&+I(X:Y|B)+I(A:B)-I(X:Y)\geq\\
    &\geq-\min\set{I(X:Y),\frac12(H(X|Y)+H(Y|X))}
      + O(\log H(XY))
  \end{split}
\end{equation*}
The latter inequality is a Shannon-type inequality if the error term
$O(\log H(XY))$ is ignored.

Thus we have two inequalities, the MMRV inequality and the spectral
inequality that are strongest at the ``different ends of the
spectrum''. For small $\lambda_{2}$ the spectral inequality gives the
strong bounds on the Ingleton expression, while the MMRV inequality
becomes almost Shannon-type, while for large values of $\lambda_{2}$
they switch their roles.

\subsection{Zhen Zhang’s question on approximate representation of the
  mutual information}
As another application we study an approximate version of a question
posed by Zhen Zhang, \cite{zhang3new}, on the characterization of
pairs with maximally non-extractable mutual information. More
specifically, we say that a pair $(X,Y)$ has \emph{maximally
  non-extractable mutual information} if the entanglement is equal to
its maximal possible value
\[
  \min\set{H(X|Y),H(Y|X),I(X:Y)}
\]
(In \cite{zhang3new} these pairs are referred to as pairs whose mutual
information is \emph{not approximately representable}).  We introduce
a relaxed version of this notion and prove the following theorem,
Theorem~\ref{th:zhang-region}.
\begin{theoremA}\label{th:intro-zhang}
  For any $\delta>0$ the set of projectivised entropy profiles
  \[
    \Big[H(X):H(Y):I(X:Y)\Big]
  \]
  of pairs $(X,Y)$ satisfying
  \[
    \mmrv(X,Y)\geq\min\set{H(X|Y),H(Y|X),I(X:Y)}-\delta\cdot H(XY)
  \]
  is dense in the region determined by the Shannon inequalities
  \[
    \Big\{[x:y:z]\st 0\leq z\leq\min\set{x,y}\Big\}\subset\Rbb P^{2}
  \]
  \shiftup
\end{theoremA}

\subsection{Structure of the article}

Section~\ref{s:prelim} is mostly devoted to the introduction of some
notion from convex analysis: the support function, lower envelope of
subsets in a Euclidean space and Legendre--Fenchel transform (convex
conjugate) of convex functions.
 
In Section~\ref{s:shape} we introduce the extension profile and shape
function of random pairs and study their elementary properties,
including the upper and lower bounds. We also prove the additivity and
the chain rule for the shape function in this section.

Graphs and their relation to random pairs are the subject of
Section~\ref{s:graphs}. We use the Expander Mixing Lemma to derive
restrictions on the sizes of subgraphs.

In the next Section~\ref{s:relation} we relate the sizes of subgraphs
of a graph and the extension profile and the shape function of a
random pair supported on the graph.

As a culmination of the theory developed so far we prove our main
result --- the spectral bound on the shape function. This is done in
Section~\ref{s:spectral}.

Finally, Section~\ref{s:applications} is devoted to the applications
of the spectral bound, where
Theorems~\ref{th:shape-minmax},~\ref{th:main-intro} and~\ref{th:intro-zhang}
are proven.

\section{Preliminaries}\label{s:prelim}
\subsection{Notation and conventions}\label{s:notations}
For $n\in\Nbb$ we denote $[n]:=\set{0,\dots,n-1}$ and we use the notation
$2^{I}$ for the power set of the set $I$.

\subsubsection{Random variables}
All our random variables are supported on finite sets.
It is important for us to make a distinction between tuples of random
variables and their joints. In order to facilitate this distinction we
denote tuples of random variables by comma-separated lists,
e.g. $(X,Y,Z)$ or $(X_{i}\st i\in[n])$, and the joint random variable
by either concatenating the letters, or using a subset of indices as a
subscript, e.g. $XY$ or $X_{I}$, $I\subset[n]$.
In addition to the usual notation for (conditional) entropy and mutual
information ($H(X)$, $H(X|Y)$, $I(X:Y)$, $I(X:Y|Z)$) we use the
following definitions
\begin{align*}
  &I(X:Y:Z):=I(X:Y)-I(X:Y|Z)\\
  &L(X,Y):=\tfrac12\big(H(X|Y)+H(Y|X)\big)\\
  &L(X,Y|Z):=\tfrac12\big(H(X|YZ)+H(Y|XZ)\big)\\
  &\ing(X,Y,A,B):=I(X:Y|A)+I(X:Y|B)+I(A:B)-I(X:Y)\\
  &\ing(X,Y,A,B|U):=I(X:Y|AU)+I(X:Y|BU)+I(A:B|U)-I(X:Y|U) 
\end{align*}
We denote a trivial random variable taking values in a one-point set
by $\bullet$ and we casually use the equivalence $\bullet Z\equiv
Z$ for any $Z$.

We use the convention that the alphabet of a random variable is
denoted by the same letter in sans-serif, that is, the  alphabet of $X$ is
$\Xsf$, etc. We use short-hand notation $(X|\zsf)$ for $X$ conditional on a value
$\zsf\in \Zsf$ of the variable $Z$, in lieu of the conventional $(X|Z=\zsf)$.

Any tuple of random variables $(X_{i}\st i\in[n])$ satisfies a series of
so called \emph{Shannon inequalities}
\[
  I(X_{I}:X_{J}|X_{K})\geq0,
  \qquad I,J,K\subset[n]
\]
where cases of empty or coinciding $I,J,K$ are also included.

We say that a (not necessarily linear) inequality for entropies of
joints $X_{I}$ is a \emph{Shannon-type} inequality if it is
syntactically implied by Shannon inequalities. In other words,
an inequality is Shannon-type if it is satisfied by the rank function of
every polymatroid.

\subsubsection{Spaces and norms}
In what follows, we consider the three-dimensional Euclidean space
$\Rbb^{3}$ with coordinates $x,y,z$. We write vectors either as row or
column vectors, whichever is more convenient to save space; the two
notations are to be treated as equivalent.

The space $\Rbb^{3}$ is equipped with the $\ell^{1}$-norm,
that is,
\[
  \|(x,y,z)\|_{1}=|x|+|y|+|z|.
\]
The dual space $(\Rbb^{3})^{*}$ carries the dual,
$\ell^{\infty}$-norm.
\[
  \|(\alpha,\beta,\gamma)\|_{\infty}=\max\set{|\alpha|,|\beta|,|\gamma|}
\]

We will also use the following shorthand
notation: for a predicate $P(x,y,z)$ we will write
\[
  \set{P(x,y,z)}:=\set{(x,y,z)\in\Rbb^{3}\st P(x,y,z)}
\]
e.g. $\set{z\geq0}$ stands for the upper half-space in $\Rbb^{3}$, etc.

\subsection{Support functions of compact convex subsets of $\Rbb^{3}$}
\label{s:support-fcn}

In this section we define \emph{support functions}, \emph{lower
  envelopes}, \emph{upper closures} of convex
sets and \emph{the Legendre--Fenchel transform (convex conjugate)} of convex
functions on convex domains. The reader is referred
to~\cite[Sections 12–16]{rockafellar1997convex} for a systematic
treatment of these notions.

We associate with every vector $(\alpha,\beta)\in\Rbb^{2}$ the
three-dimensional dual vector
\[
  l_{\alpha\beta}:=(\alpha,\beta,-1)\in(\Rbb^{3})^{*}
\]
For a compact convex set $S\subset\Rbb^{3}$ define the (partial)
\emph{support function} of $S$ as the function 
\begin{equation*}
  \Lambda_S:\Rbb^{2}\to\Rbb
\end{equation*}
equal to the
supremum on $S$ of the linear functional $l_{\alpha,\beta}$:
\[
  \Lambda_S(\alpha,\beta)
  :=\sup_{\xbf\in S} l_{\alpha,\beta}(\xbf)
  =\sup\set{\alpha\cdot x+\beta\cdot y -z\st
    (x,y,z)\in S}
\]
When convenient, especially when the set $S$ is given by a long
expression, we also write $\Lambda(S)$ for the partial support function
$\Lambda_S$.

For two compact convex sets $S,S'\subset\Rbb^{3}$ denote by
$S\oplus S'$ their Minkowski sum, by $\lambda\cdot S$ the homothetic
scaling of $S$ by the factor $\lambda\geq0$, and by $\dist(S,S')$ the
Hausdorff distance induced by the $\ell^{1}$-norm on
$\Rbb^{3}$. 
Further, define \emph{the upper closure} $S^{\uparrow}$ of
$S$ as the union of all vertical rays pointing upward ($z\to+\infty$)
and starting from points of $S$
\begin{equation*}
  S^{\uparrow}:=S\oplus\set{x=y=0,\, z\geq0}, 
\end{equation*}
and its \emph{lower envelope} as the function defined on the image
$\underline S$ of
the projection of $S$ to the horizontal plane:
\begin{equation*}
  \begin{aligned}
    \tau_{S}:\underline{S} &\to \Rbb\\
    (x,y)&\mapsto\inf\set{z\st (x,y,z)\in S}.
  \end{aligned}
\end{equation*}

\begin{remark}\ 
  \begin{enumerate}[label=(\roman*)]
  \item Since we assume compactness of the set $S$, the supremum in
    the definition above can be replaced by the maximum. We will call
    any point in $\arg\max_{s\in S} l_{\alpha\beta}(s)$ \emph{the
      point of support} by a plane with slope $(\alpha,\beta)$.
  \item Similarly, the infimum in the definition of the lower envelope
    can be replaced by a minimum. 
  \item The graph of the lower envelope is the lower part of the
    boundary of $S$.
    
    The lower envelope can equivalently be defined by
    \begin{equation*}
      \tau_{S}(x,y):=\inf\set{z\st (x,y,z)\in S^{\uparrow}}.
    \end{equation*}
    In fact, proper convex functions on convex domains are in
    one-to-one correspondence with upper-closed convex bodies in
    $\Rbb^{3}$.
  \item The partial support function of $S$ is  the
    Legendre--Fenchel transform (convex conjugate) of the lower
    envelope $\tau_{S}$ in the sense that
    \[
      \Lambda_S(\alpha,\beta)
      =
      \tau^{*}_{S}(\alpha,\beta)
      :=
      \sup_{(x,y)\in\underline S}
      \bigl(\alpha\cdot x+\beta\cdot y-\tau_S(x,y)\bigr).
    \]
    Since the Legendre--Fenchel transform is an involution on the class of
    convex functions defined on convex subsets of $\Rbb^{2}$, we also
    have
    \[
      \tau_{S}(x,y)=\Lambda_S^{*}(x,y)
    \]
  \end{enumerate}
  \shiftup
\end{remark}
For a compact set $K\subset\Rbb^{2}$ and a continuous function
$f:\Rbb^{2}\to\Rbb$, we write
\[
  \|f\|_{K}:=\sup_{(\alpha,\beta)\in K}|f(\alpha,\beta)|=\|f|_{K}\|_{\infty}
\]
and define the $\ell^{\infty}$-\emph{radius} of $K$ by
\[
  \rho_{K}:=\sup\set{\|(\alpha,\beta,-1)\|_{\infty}\st (\alpha,\beta)\in K}
\]

\begin{proposition}\label{p:legendre}
  The correspondence $S\mapsto \Lambda_{S}$ is a monotone, homogeneous,
  locally Lipschitz, max-additive with respect to unions, homomorphism
  from the class of compact convex subsets of $\Rbb^{3}$ with
  Minkowski addition and Hausdorff distance to the space of convex
  functions on $\Rbb^{2}$ with the family of local uniform norms.

  \medskip\par\noindent
  More precisely, for any pair of compact convex subsets
  $S,S'\subset\Rbb^{3}$ and any $\alpha,\alpha',\beta,\beta'\in\Rbb$
  the following properties hold
  \begin{enumerate}[label=(\roman*)]
  \item\label{i:convex}
    The function $\Lambda_{S}:\Rbb^{2}\to\Rbb$ is convex: for any
    $\lambda\in[0,1]$ the following inequality holds
    \begin{equation*}
      \Lambda_{S}(\alpha'',\beta'')
      \leq
      \lambda\cdot \Lambda_{S}(\alpha,\beta)+(1-\lambda)\cdot \Lambda_{S}(\alpha',\beta')
    \end{equation*}
    where
    $\vect{\alpha''\\\beta''}=\lambda\vect{\alpha\\\beta}+
    (1-\lambda)\vect{\alpha'\\\beta'}$.
  \item\label{i:monotone}
    $S\mapsto\Lambda_{S}$ is monotone: for $S'\subset S$ 
    \begin{equation*}
      \Lambda_{S'}(\alpha,\beta)\leq \Lambda_{S}(\alpha,\beta)
    \end{equation*}
  \item\label{i:homo}
    $S\mapsto\Lambda_{S}$ is a homomorphism:
    \begin{equation*}
      \Lambda_{S\oplus S'}=\Lambda_{S}+\Lambda_{S'}
    \end{equation*}
  \item\label{i:homogen}
    $S\mapsto\Lambda_{S}$ is homogeneous: for any $\lambda\geq0$
    \begin{equation*}
      \Lambda_{\lambda\cdot S}=\lambda\cdot \Lambda_{S}
    \end{equation*}
  \item\label{i:lip} $S\mapsto\L_{S}$ is locally Lipschitz: for every
    compact subset $K\subset\Rbb^{2}$
    \begin{equation*}
      \|\Lambda_{S}-\Lambda_{S'}\|_{K}
      \leq
      \dist(S,S')\cdot \rho_{K}
    \end{equation*}
  \item\label{i:union}
    $S\mapsto\Lambda_{S}$ is max-additive with respect to unions: 
    \begin{equation*}
      \Lambda(\convexhull (S\cup S'))(\alpha,\beta)
      =
      \max\set{\Lambda_{S}(\alpha,\beta),\Lambda_{S'}(\alpha,\beta)}
    \end{equation*}
  \end{enumerate}
  \shiftup
\end{proposition}

\begin{proof}\ 
  \begin{enumerate}[label=(\roman*)]
  \item Convexity of $\Lambda_{S}$ follows because it is a supremum of
    affine functions.
  \item Monotonicity is immediate from the definition.
  \item For Minkowski sums,
    \[
      \Lambda_{S\oplus S'}(\alpha,\beta)
      =
      \sup_{s\in S,\,s'\in S'}
      l_{\alpha,\beta}(s+s')
      =
      \Lambda_S(\alpha,\beta)+\Lambda_{S'}(\alpha,\beta).
    \]
  \item Homogeneity is proven similarly.
  \item Let $K\subset\Rbb^{2}$ be compact, $(\alpha,\beta)\in K$ and suppose
    $\dist(S,S')=\delta$. Then for any $s\in S$ there is $s'\in S'$
    such that $\|s-s'\|_{1}\leq\delta$. Let $s_{0}\in S$ be a point of
    support, and choose $s'_{0}\in S'$ with
    $\|s_{0}-s_{0}'\|_{1}\leq\delta$.  Then we have
    \begin{align*}
      \Lambda_{S}(\alpha,\beta)
      &=
        l_{\alpha,\beta}(s_{0})
        =
        l_{\alpha,\beta}(s'_{0})+l_{\alpha,\beta}(s_{0}-s'_{0})\\
      &\leq
        l_{\alpha,\beta}(s'_{0})+
        \|s_{0}-s'_{0}\|_{1}\cdot\|l_{\alpha,\beta}\|_{\infty}\\
      &\leq
        \Lambda_{S'}(\alpha,\beta)+\delta\cdot\rho_{K}
    \end{align*}
    By symmetry the required property is implied.
  \item Finally, a linear functional has the same supremum over a set and
    over its convex hull. Therefore
    \[
      \Lambda(\convexhull (S\cup S'))(\alpha,\beta)
      =
      \sup_{\xbf\in S\cup S'}l_{\alpha,\beta}(\xbf)
      =
      \max\set{\Lambda_S(\alpha,\beta),\Lambda_{S'}(\alpha,\beta)} 
    \]
  \end{enumerate}
\end{proof}
  \begin{remark}
    We note here that the upper closure $S^{\uparrow}$ of a closed
    convex subset in $\Rbb^{3}$ can be recovered from its partial
    support function. There are two ways to see that. On the one hand
    $S^{\uparrow}$ is the intersection of half-spaces in $\Rbb^{3}$:
    \begin{equation*}
      S^{\uparrow}
      =
      \bigcap_{\alpha,\beta}
      \set{z\geq \alpha\cdot x+\beta\cdot y-\Lambda_{S}(\alpha,\beta)}
    \end{equation*}

    Alternatively, $S^{\uparrow}$ is determined by the lower envelope
    $\tau_{S}$, while lower envelope $\tau_{S}$ and partial support
    function $\Lambda_{S}$ are convex conjugates of each other.
  \end{remark}

\section{Extension profile and shape function of pairs of a random
  variables}\label{s:shape}
\subsection{Entropy profile}

For a pair of random variables $(X,Y)$ define their
\emph{entropy profile} $e(X,Y)$ as the point in $\Rbb^{3}$
given by
\begin{equation*}
  e(X,Y)=\vect{H(X)\\H(Y)\\I(X:Y)}\in\Rbb^{3}
\end{equation*}
Likewise the \emph{conditional entropy profile} of the pair $(X,Y)$
for an extension $(X,Y,W)$ is
\begin{equation*}
  e(X,Y|W)=\vect{H(X|W)\\H(Y|W)\\I(X:Y|W)}\in\Rbb^{3}
\end{equation*}
We recall here that in our settings the space $\Rbb^{3}$ is equipped with the
$\ell^{1}$-norm and coordinates $(x,y,z)$.

\subsection{Extension profile}
An \emph{extension profile} of the pair $(X,Y)$ is the set of all
conditional entropy profiles of $(X,Y)$ over all extensions $(X,Y,W)$.
\begin{equation*}
  \ext(X,Y)
  :=
  \set{e(X,Y|W)\st \text{$(X,Y,W)$
      is an extension of $(X,Y)$}}\subset\Rbb^{3}
\end{equation*}
This construction is closely related to the tension region; see, for
example,~\cite{prabhakaran2014assisted,li2017extended,csirmaz2023short}.

For a triple of random variables $(X,Y,Z)$ we define the conditional
extension profile $\ext(X,Y|Z)$ in two equivalent ways.
\begin{enumerate}[label=(\roman*)]
\item Denote by $\Zsf$ the support of the extending variable $Z$, by
  $p_{Z}$ the distribution of $Z$, and define
  \begin{equation*}
    \ext(X,Y|Z):=\bigoplus_{\zsf\in\Zsf}p_{Z}(\zsf)\ext(X,Y|\zsf)
  \end{equation*}
\item
  Alternatively define
  \begin{equation*}
    \ext(X,Y|Z):=\set{e(X,Y|ZW)\st(X,Y,Z,W)\text{ is an extension of }(X,Y,Z)}
  \end{equation*}
\end{enumerate}
The equivalence of these definitions is straightforward.

\subsubsection{Extension profile is compact and convex}
These facts are fairly standard, see \cite{csiszar2011information}. 
We give a proof for completeness.
\begin{proposition}\label{p:ext-compact-convex}
  For any pair $(X,Y)$ the extension profile $\ext(X,Y)$ is a
  compact convex subset of $\Rbb^{3}$.
\end{proposition}
\begin{proof}
  To prove compactness observe that
  it suffices, by Carath\'eodory's theorem, to consider extensions
  $(X,Y,W)$ for which the alphabet size of $W$ is bounded in terms of
  the alphabet size of $(X,Y)$. Thus, the extension profile is a continuous
  image of a closed subset of a finite-dimensional simplex and
  therefore it is compact.

  Now we prove convexity.
  Let $\xbf,\xbf'\in\ext(X,Y)$. Let $(X,Y,W)$ and $(X,Y,W')$ be two
  extensions such that
  \begin{align*}
    \xbf
    &=
      e(X,Y|W)
    =
      \vect{H(X|W)\\H(Y|W)\\I(X:Y|W)}
    \quad\text{and}\\
    \xbf'
    &=
      e(X,Y|W')
      =
      \vect{H(X|W')\\H(Y|W')\\I(X:Y|W')}
  \end{align*}
  We assume that the alphabets $\Wsf$ and $\Wsf'$ of $W$ and $W'$,
  respectively, are disjoint. Denote by $p$ and $p'$ the distributions
  of $(X,Y,W)$ and $(X,Y,W')$ respectively.

  For $\lambda\in[0,1]$ define a new
  extension $(X,Y,W_{\lambda})$, where the alphabet of $W_{\lambda}$
  is $\Wsf\sqcup\Wsf'$. The distribution of $(X,Y,W_{\lambda})$ is
  defined by
  \begin{equation*}
    p_{\lambda}(\xsf,\ysf,\wsf):=
    \begin{cases}
      \lambda\cdot p(\xsf,\ysf,\wsf)&\text{if $\wsf\in\Wsf$}\\
      (1-\lambda)\cdot p'(\xsf,\ysf,\wsf)&\text{if $\wsf\in\Wsf'$}\\
    \end{cases}
  \end{equation*}
  Using the chain rule we obtain  for every $A$ from the list $\{\bullet,X,Y,XY\}$
  the  equality
  \begin{equation*}
    H(AW_{\lambda})
    =
    \lambda\cdot H(AW)+(1-\lambda)H(AW')+h(\lambda),\\
  \end{equation*}
where $h(\lambda)$ is the entropy of the
  binary variable with the distribution $(\lambda,1-\lambda)$.
  For $A=\bullet$ this inequality reduces to 
  \[
  H(W_\lambda) =  \lambda\cdot H(W)+(1-\lambda)H(W')+h(\lambda).
  \]
  Therefore, 
  \begin{equation*}
    H(A|W_{\lambda}) =  H(A W_{\lambda}) - H(W_\lambda)  
    =     \lambda\cdot H(A|W)+(1-\lambda)\cdot H(A|W')
  \end{equation*}
  In particular,
  \begin{equation*}
    e(X,Y|W_{\lambda})
    =
    \lambda\cdot e(X,Y|W) + (1-\lambda)\cdot e(X,Y|W')
    =
    \lambda\cdot\xbf + (1-\lambda)\cdot\xbf'
  \end{equation*}
\end{proof}

\subsubsection{Bounds on the extension profile}\label{s:shape-bounds}
\begin{lemma}\label{l:ext-shannon-bounds}
  For any pair of random variables $(X,Y)$ and any triple 
  $(x,y,z)\in\ext(X,Y)$ we have
  \begin{enumerate}[label=(\roman*)]
  \item\label{i:0xX}
    $0\leq x\leq H(X)$
  \item\label{i:0yY}
    $0\leq y\leq H(Y)$
  \item\label{i:L00}
    $0\leq z\leq \min\set{x,y}\leq\min\set{H(X),H(Y)}$
  \item\label{i:L10}
    $0\leq x-z\leq H(X|Y)$
  \item\label{i:L01}
    $0\leq y-z\leq H(Y|X)$
  \end{enumerate}
\end{lemma}
Define $\ext_{\max}(X,Y)$ as a subset of
$\Rbb^{3}$ of all points satisfying the inequalities in the conclusion of
the lemma. Then the lemma asserts that for any pair $(X,Y)$
\[
  \ext(X,Y)\subset\ext_{\max}(X,Y)
\]

\begin{proof}
  After the substitution 
  \[
    x=H(X|W),
    \quad
    y=H(Y|W),
    \quad
    z=I(X:Y|W)
  \]
  we observe
  \[
    x-z=H(X|YW),\quad y-z=H(Y|XW)
  \]
  Thus the required inequalities become Shannon inequalities. 
\end{proof}
In the most non-degenerate case the set
$\ext_{\max}(X,Y)$ is a polyhedron with nine vertices.  Without loss
of generality we assume that $H(X)\geq H(Y)$. Then these nine vertices
are
\begin{align*}
  &v_{0}=\vect{H(X|Y)\\H(Y|X)\\0},
  &&v_{1}=\vect{H(X|Y)\\0\\0},
  &&v_{2}=\vect{0\\H(Y|X)\\0},
  &&v_{3}=\vect{0\\0\\0},
  \\
  &v_{4}=\vect{H(X)\\H(Y)\\I(X:Y)},
  &&v_{5}=\vect{H(X)\\I(X:Y)\\I(X:Y)},
  &&v_{6}=\vect{I(X:Y)\\H(Y)\\I(X:Y)},
  \\
  &v_{7}=\vect{H(X)\\H(Y)\\H(Y)},
  &&v_{8}=\vect{H(Y)\\H(Y)\\H(Y)}.
\end{align*}
Let us focus on the points $v_{1},\dots,v_{4}$.
They always belong to $\ext(X,Y)$ since they 
correspond to the choices $W=Y,X,XY,\bullet$, respectively.  Let
$\ext_{\min}(X,Y)$ denote the convex hull of these four points.  
Since for all $X,Y$ the set $\ext(X,Y)$ is convex, it must contain $\ext_{\min}(X,Y)$.

The next Proposition summarises the bounds that we obtained for the
extension profile $\ext(X,Y)$.
\begin{proposition}\label{p:ext-minmax}
Denote 
 \begin{align*}
   \ext_{\min}(X,Y) &:=  \convexhull\set{v_{1},\dots,v_{4}},\\
   \ext_{\max}(X,Y) &:= \convexhull\set{v_{0},\dots,v_{8}} .
 \end{align*}
 Then, for every pair $(X,Y)$ we have the inclusions
 \[
   \ext_{\min}(X,Y)\subset\ext(X,Y)\subset\ext_{\max}(X,Y).
 \]
 \shiftup
\end{proposition}
Note that the sets $\ext_{\min}(X,Y)$ and $\ext_{\max}(X,Y)$ are both
polyhedral and completely determined by the entropy profile $e(X,Y)$.

\begin{remark}
  While the vertices $v_{1},\dots,v_{4}$ always belong to
  $ \ext_{\max}(X,Y)$, the four other points $v_{5},\dots,v_{8}$ may
  lie outside the profile.  However, these points can be always
  \emph{stably approximated} in the following sense.  Denote by
  $(X^{n},Y^{n})$ a pair formed by $n$ independent copies of the pair
  $(X,Y)$. Then
  \begin{equation*}
    \set{v_{1},\dots,v_{8}}\subset \lim_{n\to\infty}\frac1n\ext(X^{n},Y^{n})  
  \end{equation*}
  where the limit is in the class of compact subsets of $\Rbb^{3}$
  with the Hausdorff metric.  This claim can be proved by applying the
  Slepian--Wolf theorem to the pair $(X^{n},Y^{n})$ and including in $W$
  suitably chosen hash values of its components.

Thus, the only vertex of $\ext_{\max}(X,Y)$ that might essentially lie outside of $\ext(X,Y)$ is $v_{0}$.
\end{remark}

We will be interested in the part of the profile that lies within the
simplex $\convexhull\set{v_{0},v_{1},v_{2},v_{4}}$.  
  This is in some sense the most important part of the profile, which
  captures substantial information-theoretic properties of a random
  pair.  This part of the profile behaves in a rather regular way,
  e.g.,  it is Minkowski additive with respect to independent sums of
pairs. In particular, it is stable with respect to taking many
independent copies of the pair. The shape function, which we define in
the next section, exactly encodes the shape of that part of the
profile, while it ignores complicated and irregular features of
  the entire extension profile.

\subsection{Shape function}
For a pair $(X,Y)$ define its \emph{shape function} as
\begin{equation*}
  \Scal(X,Y):=\Lambda(\ext(X,Y))
\end{equation*}
and \emph{conditional shape function} for an extension $(X,Y,Z)$ by
\begin{equation*}
  \Scal(X,Y|Z):=\Lambda(\ext(X,Y|Z))
\end{equation*}
Alternatively we may write
\begin{align*}
  \Scal(X,Y)
  :\!\!&=
    \sup\set{\alpha\cdot H(X|W)+\beta\cdot H(Y|W)-I(X:Y|W)}\\
  &=
    \sup\set{H(XY|W)-(1-\alpha)H(X|W)-(1-\beta)H(Y|W)}\\
  \Scal(X,Y|Z)
  :\!\!&=
    \sup\set{\alpha\cdot H(X|ZW)+\beta\cdot H(Y|ZW)-I(X:Y|ZW)}\\
  &=
    \sup\set{H(XY|ZW)-(1-\alpha)H(X|ZW)-(1-\beta)H(Y|ZW)}
\end{align*}
where the suprema are over all extensions $W$ extending the pair $(X,Y)$
in the first case and extending the triple $(X,Y,Z)$ in the second.

By Proposition~\ref{p:legendre}\ref{i:homo} and~\ref{i:homogen} we
conclude
\[
  \Scal(X,Y|Z)=\sum_{\zsf\in\Zsf}p_{Z}(\zsf)\Scal(X,Y|\zsf)
\]
where $\Zsf$ is the alphabet and $p_{Z}$ is the distribution of $Z$.

By Proposition~\ref{p:legendre}\ref{i:convex} the  shape
function $\Scal(X,Y)$ (and the conditional shape function $\Scal(X,Y|Z)$) is a
convex function on $\Rbb^{2}$ for any pair $(X,Y)$ (and, respectively, for any triple $(X,Y,Z)$).
\medskip

In the next sections we derive lower and upper bounds for the shape
function of the pair. In Section~\ref{s:applications} we investigate for
which pairs of random variables the lower and upper bounds on the
shape function are attained.

\subsection{Upper bound for the shape function}
In the next proposition we show that the shape function is piecewise
affine outside of the triangle $[0,1]^{2} \cap \set{\alpha+\beta\leq1}$
in the $(\alpha,\beta)$-plane and derive the upper bound within the triangle.
\begin{proposition}[Upper bound for the shape function]\label{p:L-upper}
  For any pair of random variables $(X,Y)$ we have
  \begin{enumerate}[label=(\roman*)]
  \item\label{i:ll-triangle}
    For any $\alpha,\beta\geq0$ such that $\alpha+\beta\leq1$
    \begin{equation*}
      \Scal(X,Y)(\alpha,\beta)\leq\alpha\cdot H(X|Y) + \beta\cdot H(Y|X)
    \end{equation*}
  \item\label{i:ur-triangle}
    For any $\alpha,\beta\geq0$ such that $\alpha+\beta\geq1$
    \begin{equation*}
      \Scal(X,Y)(\alpha,\beta)
      =\alpha\cdot H(X) + \beta\cdot H(Y)-I(X:Y)
    \end{equation*}
  \end{enumerate}
  In particular, if $\alpha+\beta = 1$, then
  $\Scal(X,Y)(\alpha,\beta) = \alpha H(X|Y)+ \beta H(Y|X)$, e.g.,
  $\Scal(X,Y)(\frac12,\frac12) = \frac12 \left(H(X|Y)+H(Y|X)\right)$.
\end{proposition}

\begin{proof}\ 
  \begin{enumerate}[label=(\roman*),leftmargin=0em,itemindent=1.5em]
  \item Note that the lower bound in
    Lemma~\ref{l:ext-shannon-bounds}\ref{i:L00} and the upper bounds
    in Lemma~\ref{l:ext-shannon-bounds}\ref{i:L10} and \ref{i:L01}
    immediately imply
    \begin{equation*}
      \Scal(X,Y)(0,0)\leq 0,
      \qquad
      \Scal(X,Y)(1,0)\leq H(X|Y),
      \qquad
      \Scal(X,Y)(0,1)\leq H(Y|X)
    \end{equation*}
    On the other hand, since $\ext_{\min}(X,Y)\subset\ext(X,Y)$ we have
    \begin{equation*}
      \Scal(X,Y)(0,0)= 0,
      \quad
      \Scal(X,Y)(1,0)= H(X|Y),
      \quad
      \Scal(X,Y)(0,1)= H(Y|X)
    \end{equation*}
    We use convexity of the shape function to derive the first conclusion
    of the proposition.
    For $\alpha+\beta\leq1$ we have the convex decomposition
    \begin{equation*}
      \vect{\alpha\\\beta}
      =
      \alpha\vect{1\\0}+\beta\vect{0\\1}+(1-\alpha-\beta)\vect{0\\0}
    \end{equation*}
    Therefore
    \begin{align*}
      \Scal(X,Y)(\alpha,\beta)
      &\leq
        \alpha\cdot\Scal(X,Y)(1,0)+\beta\cdot\Scal(X,Y)(0,1)+
        (1-\alpha-\beta)\Scal(X,Y)(0,0)\\
      &=
        \alpha\cdot H(X|Y) + \beta\cdot H(Y|X)
    \end{align*}
  \item
    To prove the second conclusion we first observe that for
    $\alpha,\beta>0$, $\alpha+\beta\geq1$ and every $W$ we have 
    \[
      I(X:Y:W)\leq I(X:W)
      \quad\text{and}\quad
      I(X:Y:W)\leq I(Y:W),
    \]
    and consequently
    \[
      I(X:Y:W)\leq \alpha\cdot I(X:W)+\beta\cdot I(Y:W) 
    \]
    Now we use substitutions
    \begin{align*}
      I(X:Y:W)&=I(X:Y)-I(X:Y|W),\\
      I(X:W)&=H(X)-H(X|W),\\
      I(Y:W)&=H(Y)-H(Y|W) 
    \end{align*}
    to get
    \[
      I(X:Y)-I(X:Y|W)\leq  \alpha\bigl(H(X)-H(X|W)\bigr)+\beta\bigl(H(Y)-H(Y|W)\bigr)
    \]
    After rearranging the terms we obtain
    \[
      \alpha\cdot H(X|W) + \beta\cdot H(Y|W)-I(X:Y|W)\leq\alpha\cdot
      H(X) + \beta\cdot H(Y)-I(X:Y)
    \]
    Taking supremum over $W$'s yields
    \begin{equation*}
      \Scal(X,Y)(\alpha,\beta)\leq\alpha\cdot H(X) + \beta\cdot H(Y)-I(X:Y)
    \end{equation*}
    On the other hand, the choice $W=\bullet$ gives
    $\bigl(H(X),H(Y),I(X:Y)\bigr)\in\ext(X,Y)$, which implies the
    opposite inequality
    \[
      \Scal(X,Y)(\alpha,\beta)\geq\alpha\cdot H(X) + \beta\cdot H(Y)-I(X:Y)
    \]
    This finishes the proof of the second assertion.
  \end{enumerate}
\shiftup
\end{proof}
In a similar manner one can show that the shape function is identically zero
in the III quadrant of the $(\alpha,\beta)$-plane. In the II and IV
quadrants the stabilized shape function $\frac1n\Scal(X^{n},Y^{n})$
converges from below to a piecewise affine function with $n\to\infty$. 
We summarize the bounds on the shape function in Figure~\ref{fig:shape-upper}.
\begin{figure}[htb]
  \begin{center}
    \begin{lpic}[l(0em),draft,clean]{shape(0.7)}
      \lbl[b]{135,22;$\alpha$}
      \lbl[l]{22,108;$\beta$}
      
      \lbl{6,10;$\Scal(X,Y)=0$}
      \lbl[W]{60,60,-45;\footnotesize{\color{blue}%
          $\Scal(X,Y)=\alpha\cdot H(X)+\beta\cdot H(Y)-I(X:Y)$}}
      \lbl[W]{47,46,-45;\scriptsize{\color{magenta}%
          $\Scal(X,Y)\leq\alpha\!\cdot\! H(X|Y) + \beta\!\cdot\! H(Y|X)$}}
      \lbl[W]{55,10;\footnotesize{\color{cyan}%
          $\Scal(X,Y)=\alpha\cdot H(X|Y)$}}
      \lbl[Wr]{105,3,-90;\footnotesize{\color{orange}%
          $\frac1n\Scal(X^{n},Y^{n})
          \!\!\underset{\tiny n\to\infty}{\to}\!\!
          \alpha\!\cdot\!H(X)\!+\!\beta\!\cdot\!I(X\!:\!Y)\!-\!I(X\!:\!Y)$}}
      \lbl[W]{12,55,-90;\footnotesize{\color{cyan}%
          $\Scal(X,Y)=\beta\cdot H(Y|X)$}}
      \lbl[lW]{10,100;\footnotesize{\color{orange}%
          $\frac1n\Scal(X^{n},Y^{n})
          \!\!\underset{\tiny n\to\infty}{\to}\!\!
          \alpha\!\cdot\!I(X\!:\!Y)\!+\!\beta\!\cdot\!H(Y)\!-\!I(X\!:\!Y)$}}
    \end{lpic}
    \caption{The shape function is convex on the plane and
      stably piecewise affine outside of the
      triangle $\set{\alpha,\beta\geq0,\;\alpha+\beta\leq1}$. On the
      triangle it satisfies the upper bound from
      Proposition~\ref{p:L-upper}.}\label{fig:shape-upper}
  \end{center}
\end{figure}

 From now on, we restrict
$\Scal(X,Y)$ to $[0,1]^2$ where the shape function is additive,
  see Proposition~\ref{p:shape-additive}.  With this definition, by
Proposition~\ref{p:legendre}\ref{i:lip}, we have 1-Lipschitz
dependence of the shape function on the extension profile of the pair.
\begin{corollary}\label{cor:lipschitz}
  For any two pairs $(X,Y)$ and $(X',Y')$ we have
  \[
    \|\Scal(X,Y)-\Scal(X',Y')\|_{\infty}
    \leq
    \dist\bigl(\ext(X,Y),\ext(X',Y')\bigr)
  \]
  \shiftup
\end{corollary}

By Proposition~\ref{p:L-upper}\ref{i:ur-triangle} function
$\Scal(X,Y)$ is also affine on the upper-right triangle
$\set{\alpha+\beta\geq1}\cap[0,1]^{2}$.  The interesting behavior, not
determined by the entropy profile of the pair, occurs on the
lower-left triangle $\set{\alpha+\beta\leq1}\cap[0,1]^{2}$, where the
function is upper bounded by $\alpha\cdot H(X|Y)+\beta\cdot H(Y|X)$,
by Proposition~\ref{p:L-upper}\ref{i:ll-triangle}.  It is also worth
noting that if this upper bound is attained at some interior point of
the lower-left triangle, then convexity forces the function to
coincide with the upper bound on the whole triangle.

For every pair $(X,Y)$, we define the function $\Scal_{\max}(X,Y)$ on
the square $[0,1]^{2}$ as
\[
  \Scal_{\max}(X,Y)(\alpha,\beta)
  :=
  \max\set{
    \begin{aligned}
      &\alpha\cdot H(X)+\beta \cdot H(Y)-I(X:Y),\\
      &\alpha\cdot H(X|Y) + \beta\!\cdot\!H(Y|X)
    \end{aligned}
    }
\]
Proposition~\ref{p:L-upper} asserts that for any pair $(X,Y)$ and
any $\alpha,\beta\in[0,1]$
\[
  \Scal(X,Y)(\alpha,\beta)\leq\Scal_{\max}(X,Y)(\alpha,\beta)
\]

\subsection{Lower bound for the shape function}
\label{s:shape-lower-bound}
As we have discussed above, for every random pair $(X,Y)$, the vertices
$v_{1},v_{2},v_{4}$ belong to the extension profile. This gives a
lower bound on the shape. For every random pair $(X,Y)$ and every $\alpha,\beta\in[0,1]$
\[
  \Scal(X,Y)(\alpha,\beta)
  \geq
  \Scal_{\min}(X,Y)(\alpha,\beta)
  :=
  \max\set{
    \begin{aligned}
      &\alpha H(X)+\beta H(Y)-I(X:Y),\\
      &\alpha  H(X|Y),\\
      &\beta  H(Y|X)
    \end{aligned}
  }
\]
Equality in the bound above is equivalent to $\ext(X,Y)$
containing the triangle spanned by $\set{v_{1},v_{2},v_{4}}$ as a
face. We call such pairs $(X,Y)$ \emph{extension rigid}.

In the non-degenerate case, the apex point $(\alpha_{0},\beta_{0})$,
where three affine pieces of $\Scal_{\min}(X,Y)$ meet, is
\[
  \alpha_{0}=\frac{H(Y|X)\cdot I(X:Y)}{H(X)\cdot H(Y)-I(X:Y)^{2}},
  \quad
  \beta_{0}=\frac{H(X|Y)\cdot I(X:Y)}{H(X)\cdot H(Y)-I(X:Y)^{2}}
\]
and the value at this point is
\[
  \Scal_{\min}(X,Y)(\alpha_{0},\beta_{0})=
  \frac{H(X|Y)\cdot H(Y|X)\cdot I(X:Y)}{H(X)\cdot H(Y)-I(X:Y)^{2}}
\]
The function $\Scal_{\min}(X,Y)$ is shown in Figure~\ref{fig:lower}.
\begin{figure}[htb]
  \begin{center}
    \begin{lpic}[l(0em),draft,clean]{shape-lower(1)}
      \lbl[b]{75,7;$\alpha$}
      \lbl[l]{7,76;$\beta$}
      
      \lbl[t]{25,4.5;$\alpha_{0}$}
      \lbl[r]{5,20;$\beta_{0}$}

      \lbl[W]{40,40,-45;
        {\color{blue}%
          $\alpha\cdot H(X)+\beta\cdot H(Y)-I(X:Y)$}}
      \lbl[W]{37,9;{\color{green}{$\alpha\cdot H(X|Y)$}}}
      \lbl[W]{10,35,-90;{\color{cyan}$\beta\cdot H(Y|X)$}}
    \end{lpic}
    \caption{Function $\Scal_{\min}(X,Y)$ --- lower bound for
      $\Scal(X,Y)$.}\label{fig:lower}
  \end{center}
\end{figure}
  
Thus, we have the following bounds for the shape function.
\begin{proposition}\label{p:shape-minmax}
  For every random pair $(X,Y)$ the following inequality holds:
  \begin{equation*}
    \Scal_{\min}(X,Y)\leq\Scal(X,Y)\leq\Scal_{\max}(X,Y).
  \end{equation*}
  Both $\Scal_{\min}(X,Y)$ and $\Scal_{\max}(X,Y)$ are piecewise affine
  on the sqare $[0,1]^{2}$ and depend only on the entropy profile of
  $X,Y$.
\end{proposition}

At present we do not know whether rigid pairs exist and the lower bound is
attained. However, we will show below that if $(X,Y)$ is uniformly
supported on a balanced expander, then $(X,Y)$ is a
rigid pair up to an error $O(\log H(XY))$.  Numerical experiments
indicate that the pair $(X,Y)$ corresponding to a uniformly random
incident point and line in a projective plane over a finite field is
rigid to high precision.

\subsection{Additivity and the chain rule for the shape}
\subsubsection{Additivity}
The collection of all pairs of random variables can be equipped with
the operation of addition via independent copies.
For two, not necessarily jointly distributed, random variables $X$ and
$X'$ define their \emph{independent sum} $X\oplus X'$ as the joint of $X$ and
a copy of $X'$ distributed jointly and independently with $X$. Equivalently, $X\oplus X'$ has distribution $p_X\otimes p_{X'}$.
Further, for two pairs $(X,Y)$ and $(X',Y')$ define their
\emph{independent sum} as the new pair
\[
  (X,Y)\oplus(X',Y'):=(X\oplus X',Y\oplus Y')
\]
where the copies of $(X,Y)$ and $(X',Y')$ are taken independently.

The correspondence $(X,Y)\mapsto \ext(X,Y)$ is \emph{not} a
homomorphism with respect to independent sum and Minkowski addition.
For example, the Slepian--Wolf theorem allows us to find for many
independent copies of $(X,Y)$ extensions that are not scalings of any
extensions for a single $(X,Y)$. However, the correspondence between
$(X,Y)$ and $\ext(X,Y)$ is additive on the most essential (lower) part
of the profile, recorded by the shape.

\begin{proposition}\label{p:shape-additive}
  The shape function is additive with respect to
  independent sums. Let $(X,Y)$ and $(X',Y')$ be two random pairs,
  then  
  \begin{equation*}
    \Scal\big((X,Y)\oplus(X',Y')\big)=\Scal(X,Y)+\Scal(X',Y')
  \end{equation*}
  \shiftup
\end{proposition}

\begin{proof}
  It will be convenient to use the following equivalent form of the
  definition of the extension support function:
  \[
    \Scal(X,Y)(\alpha,\beta)
    =
    \sup_{W}\set{H(XY|W)-(1-\alpha)H(X|W)-(1-\beta)H(Y|W)}
  \]
  For any extensions $(X,Y,W)$ and $(X',Y',W')$ we can take 
  the extension $(X\oplus X',Y\oplus Y',W\oplus W')$.
  Since  we have
  \[
    e(X,Y|W) + e(X',Y'|W')
    =
    e\big((X,Y)\oplus(X',Y')|W\oplus W'\big)
  \]
  it follows that
  \[
    \ext(X,Y)\oplus\ext(X',Y')\subset\ext\big((X,Y)\oplus(X',Y')\big)
  \]
  Therefore, by additivity and monotonicity of the support function,
  Proposition~\ref{p:legendre}\ref{i:monotone} and~\ref{i:homo}, we have
  \[
    \Scal(X,Y)+\Scal(X',Y')\leq \Scal\big((X,Y)\oplus(X',Y')\big)
  \]
  It remains to prove the opposite inequality.
  Let
  \[
    (X'',Y''):=(X,Y)\oplus (X',Y')
  \]
  Fix $\alpha,\beta\in[0,1]$ and let an extension $(X'',Y'',W'')$ be
  such that
  \begin{align*}
    \Scal(X'',Y'')&(\alpha,\beta)
    =l_{\alpha,\beta}(\xbf'')\\
    &=
      H(X''Y''|W'')-(1-\alpha) H(X''|W'')-(1-\beta) H(Y''|W'')
  \end{align*}
  Define $W:=W''X'Y'$ and $W':=W''$. Then
  \begin{align*}
    H(X''Y''|W'')
    &=
      H(XY\oplus X'Y'|W'')\\
    &=
      H(X'Y'|W'')+H(X,Y|W''X'Y')\\
    &=H(X'Y'|W')+H(X,Y|W)
  \end{align*}
  Also
  \begin{align*}
    H(X''|W'')
    &=
      H(X\oplus X'|W'')
    =
      H(X'|W'')+H(X|W''X')\\
    &\geq
      H(X'|W'')+H(X|W''X'Y')
    =
      H(X'|W')+H(X|W)
  \end{align*}
  Similarly
  \begin{equation*}
    H(Y''|W'')\geq H(Y'|W')+H(Y|W)
  \end{equation*}
  Therefore for any $0\leq\alpha,\beta\leq1$
  \begin{align*}
    \Scal(X'',Y'')(\alpha,\beta)
    &=
      H(X''Y''|W'')-(1-\alpha) H(X''|W'')-(1-\beta) H(Y''|W'')\\
    &\leq
    H(XY|W)-(1-\alpha) H(X|W)-(1-\beta) H(Y|W)\\
    &\quad
      +H(X'Y'|W')-(1-\alpha) H(X'|W')-(1-\beta) H(Y'|W')\\
    &\leq
    \Scal(X,Y)(\alpha,\beta)+\Scal(X',Y')(\alpha,\beta)
  \end{align*}
  \shiftup
\end{proof}

\subsubsection{Chain rule}
We regard the identities in the next proposition as the \emph{chain rule for
the extension profile and the shape}.

\begin{proposition}[Chain rule for the extension and the shape]
  For any pair $(X,Y)$  and any extension $(X,Y,Z)$ it holds
  \begin{align*}
    \ext(XZ,YZ)&=\ext(X,Y|Z) \oplus \ext(Z,Z)\\
    \Scal(XZ,YZ)&=\Scal(X,Y|Z) + \Scal(Z,Z)
  \end{align*}
  \shiftup
\end{proposition}

\begin{proof}
  We will prove the statement on the level of extension profiles; the
  identity for shapes then follows from
  additivity of the partial support function with respect to Minkowski
  sums, Proposition~\ref{p:legendre}\ref{i:homo}.  Observe that
  \begin{equation*}
    H(AZ|W)=H(A|ZW)+H(Z|W),
    \quad\text{for each A in the list $\set{X,Y,XY}$}
  \end{equation*}
  Therefore, for any extension $(XZ,YZ,W)$
  \begin{equation*}
    e(XZ,YZ|W)
    =
    e(X,Y|ZW) + e(Z,Z|W)
  \end{equation*}
  The summands in the right-hand side belong to $\ext(X,Y|Z)$ and $\ext(Z,Z)$,
  respectively. Therefore
  \begin{equation*}
    \ext(XZ,YZ)\subset\ext(X,Y|Z)\oplus\ext(Z,Z)
  \end{equation*}
  To prove the opposite inclusion we will need the following lemma.
  
  \begin{lemma}
    Let $x$ be a point in $\ext(X,Y|Z)$ realized by some extension
    $(X,Y,Z,U)$.  Then there is another extension $(X,Y,Z,U')$
    realizing the same point $x$ and such that \( I(Z:U') = 0\).
  \end{lemma}
  \begin{proof}
    Without loss of generality we assume that $Z$ is distributed on
    $[k]$ for some integer $k$.  Let us consider $k$ random variables
    $U_i$ for $i=0,\ldots, k-1$ such that for every $i$ the
    distribution of $U_i$ is the same as the conditional distribution
    $(U | Z=i)$. Now we let
    \[
      U'  := (U_0,…,U_{k-1})
    \]
    i.e., the $i$-th marginal of $U'$ is equivalent to $U_i$ (for the
    rest of the argument, it is not important how all $U_i$ are
    correlated with each other; we may assume for example that all
    $U_i$ are sampled independently).
    
    Now we couple $U'$ with $Z$ as two independent random variables,
    and define the conditional distribution of $(X,Y)$ given each pair
    of values of $(U',Z)$ as follows: we let
    \[
      \prob[ (X,Y) = (x,y) \ |\  U'= (u_0,\ldots, u_{k-1})\ \text{and}\ Z=j]
    \]
    be equal to
    \( \prob[ (X,Y) = (x,y) \ |\ U =u_j \text{ and }Z=j].\) For this
    definition we get immediately $H(XY | ZU') = H(XY | ZU)$, and
    similarly $H(X | ZU') = H(X | ZU)$ and $H(Y | ZU') = H(Y | ZU)$.
  \end{proof}  
  
  Let us consider arbitrary extensions $(X,Y,Z,U)$ and $(Z,Z,V)$.  Due
  to the Lemma, we may assume that the extension $(X,Y,Z,U)$ satisfies
  $U \indep Z$.  Now we use the adhesivity property of entropic
  polymatroids over the common variable $Z$,
  see~\cite{matus2007infinitely}, which implies that we may
  realize an extension $(X,Y,Z,U,V)$ such that $XYU\indep V|Z$.  Then,
  for each $A$ in the list $\set{X,Y,XY}$, we have 
  \begin{align*}
    H(AZ|UV)
    &=
      H(Z|UV)  +  H(A|ZUV) \\
    &=
      H(Z|V) - I(U:Z|V)+ H(A|ZU) - I(A:V|ZU) \\
    &=
      H(Z|V) + H(A|ZU).
  \end{align*}
  Here  $I(A:V|ZU) = 0$ follows from  $AU\indep V|Z$,
  and $I(U:Z|V) = 0$ follows from the combination of $ U \indep Z$ and
  $U\indep V|Z$. 
  
  Consequently,
  \[
    e(X,Y|ZU)+e(Z,Z|V)
    =
    e(XZ,YZ|UV)
    \in
    \ext(XZ,YZ)
  \]
  and
  \begin{equation*}
    \ext(X,Y|Z)\oplus\ext(Z,Z)\subset\ext(XZ,YZ)
  \end{equation*}
  This finishes the proof for the chain rule for extension
  profiles. Applying the support function operator and using its
  additivity, Proposition~\ref{p:legendre}\ref{i:homo}, we obtain the
  chain rule for the shape function.
\end{proof}

\subsection{Examples: a few simple cases}
Below we list shape functions for some simple pairs. We would like to remind the reader 
here that we restrict shape function on the square $[0,1]^{2}$ in
$(\alpha,\beta)$-plane. 

\begin{enumerate}
\item 
  For a pair of the form $(X,X)$ we have
  \[
    \Scal(X,X)(\alpha,\beta)=\max\set{0, (\alpha+\beta-1) H(X)}
  \]
\item
  For $(X,\bullet)$ we have
  \[
    \Scal(X,\bullet)(\alpha,\beta)=\alpha\cdot H(X)
  \]
\item
  If $(X,Y)$ is such that $X\indep Y$, then
  \[
    (X,Y) = (X,\bullet)\oplus(\bullet,Y)
  \]
  Therefore
  \[
    \Scal(X,Y)(\alpha,\beta)=\alpha\cdot H(X)+\beta\cdot H(Y)
  \]
\item
  If $(X,Y)$ is such that $H(X|Y)=0$, then
  \[
    \Scal(X,Y)(\alpha,\beta)
    =
    \max\set{
      \beta H(Y|X),\,
      (\alpha-1)H(X)+\beta H(Y)
    }
  \]
\end{enumerate}
In each of these cases, Shannon inequalities determine the shape function
completely from the entropy profile of the pair. The equalities above
can be easily derived directly from the definitions.

\section{Extension profiles of bipartite graphs}\label{s:graphs}
\subsection{Graphs and notation}
\label{s:graphs-notations}
For a biregular bipartite graph $\Gsf=(\Xsf\sqcup\Ysf;\Esf)$ and its
subgraph $\Hsf$, denote by $\Xsf_{\Hsf}$, $\Ysf_{\Hsf}$ and
$\Esf_{\Hsf}$ the left part, the right part, and the edge-set of $\Hsf$,
respectively. By a subgraph we will always mean a non-empty subgraph
without isolated vertices. In particular, both the left and right parts will be
non-empty. We also use the following notation
\begin{align*}
  [\Xsf]
  &:=\log|\Xsf|
  &[\Xsf_{\Hsf}]
  &:=\log|\Xsf_{\Hsf}|\\
  [\Ysf]
  &:=\log|\Ysf|
  &[\Ysf_{\Hsf}]
  &:=\log|\Ysf_{\Hsf}|\\
  [\Esf]
  &:=\log|\Esf|
  &[\Esf_{\Hsf}]
  &:=\log|\Esf_{\Hsf}|\\
  [\Xsf:\Ysf]
  &:=[\Xsf]+[\Ysf]-[\Esf]
  &[\Xsf_{\Hsf}:\Ysf_{\Hsf}]
  &:=[\Xsf_{\Hsf}]+[\Ysf_{\Hsf}]-[\Esf_{\Hsf}]
\end{align*}
By $d_{1}(\Gsf)$ and $d_{2}(\Gsf)$ we denote the left and right
degrees of $\Gsf$. By $\lambda_{i}(\Gsf)$ we denote the eigenvalues of
$\Gsf$ indexed from 1 in the decreasing order and counted with
multiplicities. That is, $\lambda_{1}(\Gsf)$ is the largest eigenvalue
and $\lambda_{2}(\Gsf)$ is the second largest. Then, due to
biregularity of the graph, we have
\begin{align*}
  \log d_{1}(\Gsf)
  &=
    [\Ysf|\Xsf]=[\Esf]-[\Xsf]\\
  \log d_{2}(\Gsf)
  &=[\Xsf|\Ysf]=[\Esf]-[\Ysf]\\
  \log\lambda_{1}(\Gsf)
  &=
    \tfrac12[\Xsf|\Ysf]+\tfrac12[\Ysf|\Xsf]
\end{align*}
We call a biregular bipartite graph $\Gsf$ \emph{admissible} if it is
sufficiently large and not too close to a complete bipartite graph, in
the sense that
\[
  [\Esf]\geq {\log 3} 
    \quad\text{and}\quad
  [\Xsf:\Ysf]  \red{>}\log 2
\]
The constraint $ [\Esf] = \log |\Esf| \geq \log 3>1$ means that $\Esf$
contains at least three elements; in what follows we use the fact that
$\log [\Esf]\ge \log \log 3 >0$. 
The density constrained is necessary to formulate Proposition~\ref{p:lower_bound_lambda_2}.

\subsection{Expanders}\label{s:expanders}
The second largest eigenvalue of
  an admissible graph can not be arbitrarily small. The following
  proposition follows from~\cite[corollary~4]{hoholdt2012eigenvalues}.
  \begin{proposition}\label{p:lower_bound_lambda_2}
    For any admissible graph
    $G:=(\Xsf\sqcup\Ysf,\Esf)$  it holds
    \begin{equation*}
      2\log\lambda_{2}
      \geq
      \log\max\set{d_{1},d_{2}}- \log 2
    \end{equation*}
    where $\lambda_{2}$ is the second largest eigenvalue and $d_{i}$'s
    are the degrees of $\Gsf$.
  \end{proposition}

  In what follows, we adapt the definition of an unbalanced bipartite
  expander to our setting.  Note that, unlike in the standard
  definition of an expander, we do not require the vertex degrees to
  be bounded by a constant.

We say that an infinite family of graphs $\Gsf$ is a \emph{right-heavy
    expander family} if there is a constant $C\geq0$ such that
  \[
    2\log\lambda_{2}\leq\log d_{1}+C
  \]
  and define \emph{a left-heavy expander family} in analogous manner,
 by requiring that there is a constant $C\geq0$ such that
  \[
    2\log\lambda_{2}\leq\log d_{2}+C
  \]

  We say that a family of graphs is a \emph{balanced expander family} if
  it is simultaneously a right-heavy and a left-heavy expander
  family. 
  For balanced expanders this property can be equivalently formulated as follows:
  there exists $C\geq0$
  such that
  \[
    2\log\lambda_{2}\leq\log\lambda_{1}+C.
  \]
  For balanced expanders it holds
  \[
    \Bigl|[X]-[Y]\Bigr|\leq 2C
    \quad\text{and}\quad
    \Bigl|[X|Y]-[Y|X]\Bigr|\leq 2C
  \]
  When we say that $\Gsf$ is a \emph{left-heavy/right-heavy/balanced
    expander} we implicitly assume that it belongs to some expander
  family with a fixed constant $C$ and a universal quantifier is applied
  to $\Gsf$.

\subsection{Extension profile of an admissible graph}
For a subgraph $\Hsf$ of an admissible graph $\Gsf$ we define
\begin{equation*}
  e(\Gsf|\Hsf)
  :=
  \vect{[\Xsf_{\Hsf}]\\{}[\Ysf_{\Hsf}]\\{}[\Xsf_{\Hsf}:\Ysf_{\Hsf}]}
  \in\Rbb^{3}
\end{equation*}
In particular, if $H$ consists of a single edge, then
$e(\Gsf|\Hsf) = (0,0,0)$.  As usual the space $\Rbb^{3}$ is equipped
with the $\ell^{1}$-norm and $(x,y,z)$-coordinates.

Define the \emph{graph extension profile} of
$\Gsf$ by
\begin{equation*}
  \extgr(\Gsf)
  :=
  \set{e(\Gsf|\Hsf)\st
    \text{$\Hsf$ is a subgraph of $\Gsf$}}\subset\Rbb^{3}   
\end{equation*}
For reasons that will become clear below it is convenient to split the
graph extension profile into two parts, the upper and the lower, that
lie above and below the threshold plane $\set{z=[X:Y]-\log2}$,
respectively:

\begin{align*}
  \extgru(\Gsf)&:=\extgr(\Gsf)\cap\set{z\geq [\Xsf:\Ysf]-\log2}\\
  \extgrl(\Gsf)&:=\extgr(\Gsf)\cap\set{z< [\Xsf:\Ysf]-\log2}
\end{align*}
Note that the inequality in the second definition is strict, so the
sets $\extgru(\Gsf)$ and $ \extgrl(\Gsf)$ are disjoint.
Besides, for all admissible graphs we have $(0,0,0)\in  \extgrl(\Gsf)$.

Obvious inequalities for cardinalities imply the following simple
lemma.
\begin{lemma}\label{l:extgr-shannon-bounds}
  Let $\Gsf$ be an admissible graph. Then for any
  $(x,y,z)\in\extgr(\Gsf)$
  \begin{enumerate}[label=(\roman*)]
  \item\label{i:gr-0xX}
    $0\leq x\leq [\Xsf]$
  \item\label{i:gr-0yY}
    $0\leq y\leq [\Ysf]$
  \item\label{i:gr-L00}
    $0\leq z\leq \min\set{x,y}\leq \min\set{[\Xsf],[\Ysf]}$
  \item \label{i:gr-L10}
    $0\leq x-z\leq \log d_{2}=[\Xsf|\Ysf]$
  \item\label{i:gr-L01}
    $0\leq y-z\leq \log d_{1}=[\Ysf|\Xsf]$
  \end{enumerate}
  \shiftup
\end{lemma}

\subsection{Expander Lemma and graph extension profile}
The following theorem is a simple corollary of the bipartite expander
mixing
lemma,~\cite{alon1988explicit,hoory2006expander,evra2015mixing}, and
is proven in~\cite{matveev2026spectral}.
We will refer to this statement as EML-$\log$.

\begin{theorem}[Expander Mixing Lemma on $\log$ scale,
  see~\cite{matveev2026spectral}]
  \label{th:mixing-log}\ \\
  For any (not necessarily induced) subgraph $\Hsf$ in a biregular
  bipartite graph $\Gsf=(\Xsf\sqcup\Ysf, \Esf)$ the following
  alternative holds. Either
  \begin{enumerate}[label=(\roman*)]
  \item\label{i:density}
    $[\Xsf_{\Hsf}:\Ysf_{\Hsf}]
    \geq [\Xsf:\Ysf] - \log2$
  \item[or]
  \item\label{i:spectral}
    $\frac12[\Xsf_{\Hsf}]+\frac12[\Ysf_{\Hsf}]-[\Xsf_{\Hsf}:\Ysf_{\Hsf}]
    \leq
    \log\lambda_{2}(\Gsf)+\log 2$
  \end{enumerate}
\end{theorem}
The statement above restricts the possible locations of points in the
graph extension profile $\extgr(\Gsf)$. For every point in
$\extgr(\Gsf)$, either
\begin{itemize}
\item the point belongs to the upper part, $\extgru(\Gsf)$, in which
  case there are no further restrictions beyond those in
  Lemma~\ref{l:extgr-shannon-bounds};
\item the point belongs to the lower part, $\extgrl(\Gsf)$, in which case it
  must satisfy the inequality in the second branch of EML-$\log$.
\end{itemize}

Next, we apply the partial support function to the upper
and lower parts of the graph extension profile and derive some bounds
on it. Since this profile is
a finite subset of $\Rbb^{3}$ and not convex, we use the convention
\[
  \Lambda(S):=\Lambda(\convexhull S)
\]
for arbitrary compact sets $S\subset\Rbb^{3}$.

\subsection{Bounds for the support function of the lower and upper parts}
\def\flo{\check f}
\def\fup{\hat f}

Note that, since all graphs under consideration are admissible, the lower graph extension profile is nonempty.
Indeed, it contains at least the point $(0,0,0)$, which corresponds to a subgraph consisting of a single edge.
\begin{lemma}\label{l:graph-support-pts}
  Let $\Gsf=(\Xsf\sqcup\Ysf,\Esf)$ be a biregular bipartite graph.
  Let
  \[
    f:=\Lambda\bigl(\extgr(\Gsf)\bigr),
    \qquad
    \flo:=\Lambda\bigl(\extgrl(\Gsf)\bigr),
    \qquad
    \fup:=\Lambda\bigl(\extgru(\Gsf)\bigr)
  \]
  denote the support functions of the extension profile of $\Gsf$ and
  its lower and upper parts, respectively.  Then
  \begin{enumerate}[label=(\roman*)]
  \item\label{i:f1001}
    $f(1,0)
    =
    [\Xsf|\Ysf]$
    and 
    $f(0,1)
    =
    [\Ysf|\Xsf]$
  \item\label{i:fup00}
    $\fup(0,0)
    \leq
    -[\Xsf:\Ysf]+\log2$
  \item\label{i:flo00}
    $\flo(0,0)=0$ 
  \item\label{i:flo-half}
    $\flo(\tfrac12,\tfrac12)
    \leq
    \log\lambda_{2}(\Gsf)+\log2$
  \end{enumerate}
\end{lemma}
Note that $\flo,\fup\leq f$ and, therefore,
Lemma~\ref{l:graph-support-pts}\ref{i:f1001} implies
\[
  \flo(1,0)\leq [\Xsf|\Ysf],
  \quad
  \flo(0,1)\leq [\Ysf|\Xsf],
  \quad
  \fup(1,0)\leq [\Xsf|\Ysf],
  \quad
  \fup(0,1)\leq [\Ysf|\Xsf],
\]
\begin{proof}
  \begin{enumerate}[label=(\roman*)]
  \item The inequalities
    \[
      f(1,0)
      \leq
      [\Xsf|\Ysf],
      \qquad
      f(0,1)
      \leq
      [\Ysf|\Xsf]
    \]
    are direct consequences of
    Lemma~\ref{l:extgr-shannon-bounds}\ref{i:gr-L10}
    and~\ref{i:gr-L01}.
    The opposite inequalities follow from the choice of subgraph
    equal to the neighborhood of a point
    $\ysf_0\in\Ysf$ in the right part. It gives the point
    \[
      \vect{[\Xsf|\Ysf]\\0\\0}\in\extgr(\Gsf)
    \]
    and the inequality
    \[
      f(1,0)
      \geq
      [\Xsf|\Ysf]
    \]
    A similar argument works for the second inequality.
  \item This inequality follows directly from the definition of the
    upper part $\extgru(\Gsf)$.
  \item The inequality $\flo(0,0)\leq0$ follows from
    non-negativity of the $z$-coordinate of every point in
    $\extgrl(\Gsf)$. The opposite inequality follows from the choice
    of a subgraph consisting of a single edge, which gives the point
    $(0,0,0)\in\extgrl(\Gsf)$.
  \item Every point in $\extgrl(\Gsf)$ must satisfy
    conclusion~\ref{i:spectral} of Theorem~\ref{th:mixing-log}.
    Therefore,
    \[
      \frac12x+\frac12y-z
      \leq
      \log\lambda_{2}(\Gsf)+\log2
    \]
    for every $(x,y,z)\in\extgrl(\Gsf)$. Taking the supremum gives
    the desired inequality.
  \end{enumerate}
\end{proof}

\begin{proposition}\label{p:graph-support-upper}
  Let $\Gsf=(\Xsf\sqcup\Ysf,\Esf)$ be a biregular bipartite
  graph. Let $\fup:=\Lambda\bigl(\extgru(\Gsf)\bigr)$ be the support
  function of the upper part of the graph extension profile of
  $\Gsf$. Then for any $\alpha,\beta\in[0,1]$, $\alpha+\beta\leq1$
  \[
    0\leq
    \fup(\alpha,\beta)
    -
    \bigl(\alpha\cdot[X]+\beta\cdot[Y]-[X:Y]\bigr)
    \leq
    \log2
  \]
\end{proposition}
\begin{proof}
  On the one hand we can take $\Gsf$ as its own subgraph. It gives
  the point
  $\bigl([\Xsf],[\Ysf],[\Xsf:\Ysf]\bigr)\in\extgru(\Gsf)$. Therefore
  \[
    \fup(\alpha,\beta)
    \geq
    \ell_{\alpha\beta}\bigl([\Xsf],[\Ysf],[\Xsf:\Ysf]\bigr)
    =
    \alpha\cdot[X]+\beta\cdot[Y]-[X:Y]
  \]
  On the other hand using convexity of the support function and the
  convex decomposition
  \[
    \vect{\alpha\\\beta}
    =
    \alpha\vect{1\\0}+\beta\vect{0\\1}+(1-\alpha-\beta)\vect{0\\0}
  \]
  we get
  \begin{align*}
    \fup(\alpha,\beta)
    &\leq
      \alpha\cdot \fup(1,0)+
      \beta\cdot  \fup(0,1)+
      (1-\alpha-\beta) \fup(0,0)\\
    &\leq
      \alpha\cdot[\Xsf|\Ysf]+\beta\cdot[\Ysf|\Xsf]
      -(1-\alpha-\beta)\bigl([\Xsf:\Ysf]-\log2\bigr)\\
    &\leq
      \alpha\cdot[\Xsf]+\beta\cdot[\Ysf]-[\Xsf:\Ysf]+\log2
  \end{align*}
  where we use Lemma~\ref{l:graph-support-pts} to bound $\fup(1,0)$,
  $\fup(0,1)$ and $\fup(0,0)$ from above.
\end{proof}

\begin{proposition}\label{p:graph-support-lower}
  Let $\Gsf=(\Xsf\sqcup\Ysf,\Esf)$ be a biregular bipartite graph with
  the second largest eigenvalue $\lambda_2$.  Let
  $\flo:=\Lambda\bigl(\extgrl(\Gsf)\bigr)$ be the support function of the
  lower part of the graph extension profile of $\Gsf$. Then
  \begin{itemize}
  \item for any $0\leq\alpha\leq\beta\leq1$, $\alpha+\beta\leq1$
    \[
      \flo(\alpha,\beta)
      \leq
      (\beta-\alpha)[\Ysf|\Xsf] + 2\alpha\log\lambda_2+\log2
    \]
  \item for any $0\leq\beta\leq\alpha\leq1$, $\alpha+\beta\leq1$
    \[
      \flo(\alpha,\beta)
      \leq
      (\alpha-\beta)[\Xsf|\Ysf] + 2\beta\log\lambda_2+\log2
    \]
  \end{itemize}
\end{proposition}
\begin{proof}
  Use convexity of the support function and the convex decomposition
  \[
    \vect{\alpha\\\beta}
    =
    (\beta-\alpha)\vect{0\\1}+2\alpha\vect{1/2\\1/2}+(1-\alpha-\beta)\vect{0\\0}
  \]
  to estimate
  \begin{align*}
    \flo(\alpha,\beta)
    &\leq
      (\beta-\alpha)\flo(0,1)+2\alpha\cdot
      \flo(\tfrac12,\tfrac12)+(1-\alpha-\beta)\flo(0,0)\\
    &\leq
      (\beta-\alpha)[\Ysf|\Xsf]+2\alpha\cdot
      \log\lambda_{2}+\log2 
  \end{align*}
  where Lemma~\ref{l:graph-support-pts} is used to bound
  $\flo(\tfrac12,\tfrac12)$, $\flo(0,1)$, $\flo(1,0)$ and $\flo(0,0)$.
  
  The second inequality is derived similarly.
\end{proof}

\section{Relation between extension profiles of graphs and pairs of
  random variables}\label{s:relation}
\subsection{Bipartite graphs and random pairs: notation and definitions}
Suppose $(X,Y)$ is a pair of random variables uniformly supported
on a graph $\Gsf=(\Xsf\sqcup\Ysf,\Esf)$, that is
\[
  \Xsf=\supp X,
  \quad
  \Ysf=\supp Y,
  \quad
  \Esf=\supp XY,
\]
and each of $X$, $Y$ and $XY$ is uniform on it's support.

In this situation we have the following identities:
\begin{align*}
  &H(X)=[\Xsf],&&H(Y)=[\Ysf],\\
  &H(XY)=[\Esf],&&I(X:Y)=[\Xsf:\Ysf],\\
  &H(Y|X)=\log d_{1}(\Gsf), &&H(X|Y)=\log d_{2}(\Gsf),\\
  &L(X,Y)=\log\lambda_{1}(\Gsf).
\end{align*}

The following definitions will also be used below. We say that a
random pair $(X',Y')$ is \emph{subsupported} on a bipartite graph
$\Gsf=(\Xsf\sqcup\Ysf,\Esf)$ if
\[
  \supp X'\subset\Xsf,
  \qquad
  \supp Y'\subset\Ysf
  \qquad
  \supp X'Y'\subset\Esf
\]
We call a random variable $X$ \emph{$\delta$-uniform} for some $\delta>1$ if
\[
  \frac{\max_{x} \prob[X=x]}{\min_{x}\prob[X=x]}\leq \delta,
\]
where the extrema are taken over $\supp(X)$.
We call a pair $(X,Y)$ \emph{$\delta$-uniform} if each of $X$, $Y$ and
the their joint $XY$ is a $\delta$-uniform random variable.

\subsection{Extension profiles of graphs and random pairs}
In this section we prove the following theorem, which relates the
extension profile of a pair $(X,Y)$ uniformly distributed on its
support to the graph extension profile of its support
$\Gsf:=\supp(X,Y)$. This theorem will be the principal tool for
deriving the further results of the article.

\begin{theorem}\label{th:entropic-in-graph}
  There exists a universal constant $C>0$ such that the following
  holds. Let $(X,Y)$ be a uniform pair supported on an admissible
  graph $\Gsf=(\Xsf\sqcup\Ysf,\Esf)$. Then
  \[
  \ext(X,Y)\subset\convexhull\bigl(\nbd_{\delta}\extgr\Gsf\bigr)
  \]
  where $\delta=C\cdot\log[\Esf]=C\cdot \log H(XY)$ and $\nbd_\delta$
  denotes the $\delta$-neighborhood with respect to the $\ell^1$-norm
  on $\Rbb^3$.
\end{theorem}

Before we prove Theorem~\ref{th:entropic-in-graph}, we need to develop
some necessary tools.

The following version of \emph{Almost-Uniform Decomposition Lemma} is
very similar to~\cite[Theorem 4.2.A]{matveev2026spectral}. It is based
on \cite[Theorem 3]{alon2007partitioning}.

\begin{lemma}[Almost Uniform Decomposition Lemma]
  \label{l:aud}
  There are universal constants $C_0 > 0$ and $\delta_0 \geq1$ such
  that for every $0<\epsilon<1$ and for any pair of random variables
  $(X,Y)$ such that $H(X,Y)\ge \log 3$, there
  exists an extension by a random variable $V$ with alphabet
  $\Vsf = \set{v_0, v_1 ,..., v_{k-1}, v_{\infty} }$ satisfying the
  following properties hold:
  \begin{enumerate}[label=(\roman*)]
  \item\label{i:smalltail4}
    $\displaystyle\Pbb[V=v_{\infty} ]  \leq \epsilon$
  \item\label{i:entropybound4}
    $H(V)\leq C_0 \cdot \log \left(\frac{H(XY)}\epsilon\right)$
  \item\label{i:regular4}
    $(X,Y|v_{i})$ is $\delta_0$-uniform for every $i\in[k]$.
  \end{enumerate}
\end{lemma}
\begin{remark}
  The condition $H(X,Y)\ge \log 3$ in the lemma can be replaced by
  $H(X,Y)\ge \log (2+\gamma)$, for any fixed $\gamma>0$.  We use the
  assumption $H(X,Y)\ge \log 3$, which is valid for $(X,Y)$ uniformly
  distributed on a support with at least three elements.
\end{remark}
 The proof of~\ref{l:aud} closely parallels the one 
 of~\cite[Theorem~4.2.A]{matveev2026spectral}, nevertheless for the sake of
  self-containedness we provide one in Appendix.

We apply this lemma with $\epsilon = 1/\log|\Esf|$. In this setting
\ref{i:entropybound4} rewrites to
\[
   H(V)\leq
   C_0 \cdot  \log \left(H(XY)\cdot\log|\Esf| \right)
   =
   O\left(  \log  \log |\Esf| \right) .
 \]

We also need the following lemma.
\begin{lemma}
  \label{l:unif-subsuported}
  Suppose $(X',Y')$ is a $\delta_{0}$-uniform pair subsupported on an
  admissible graph $\Gsf=(\Xsf\sqcup\Ysf,\Esf)$. Then
  \[
    e(X',Y')\in\nbd_{5\log\delta_{0}}\extgr(\Gsf)
  \]
\end{lemma}

\begin{proof}
  Let $\Xsf'\subset\Xsf$, $\Ysf'\subset\Ysf$ and $\Esf'\subset\Esf$ be
  the supports of $X'$, $Y'$ and $X'Y'$, respectively. Let
  $\Hsf=(\Xsf'\sqcup\Ysf', \Esf')$ be the corresponding subgraph of
  $\Gsf$. Since $(X',Y')$ is $\delta_{0}$-uniform,
  by~\cite[Lemma 2.2.B]{matveev2026spectral}, we have
  \begin{align*}
    [\Xsf']
    &\geq
      H(X')\geq [\Xsf']-\log\delta_{0}\\
    [\Ysf']
    &\geq
      H(Y')\geq [\Ysf']-\log\delta_{0}\\
    [\Esf']
    &\geq
      H(X'Y')\geq [\Esf']-\log\delta_{0}\\
  \end{align*}
  Therefore
  \[
    \bigl|I(X':Y')-[\Xsf':\Ysf']\bigr|
    \leq 3\log\delta_0
  \]
  Hence
  \[
    \|e(X',Y')-e(\Gsf|\Hsf)\|_{1}\leq 5\log\delta_{0}
  \]
  Since $e(\Gsf|\Hsf)$ is a point in $\extgr(\Gsf)$, the lemma is
  proven.
\end{proof}

\begin{proof}[Proof of Theorem~\ref{th:entropic-in-graph}]
  Suppose $(X,Y,W)$ is an extension of $(X,Y)$. Denote by $\Wsf$ the
  alphabet of $W$.  Then
  \[
    e(X,Y|W)=\sum_{\wsf\in\Wsf}p_{W}(\wsf)\cdot e(X,Y|\wsf)
  \]
  Therefore 
  \begin{equation}\label{eq:xy|w}
    e(X,Y|W)\in\convexhull\set{e(X,Y|\wsf)\st \wsf\in\Wsf}
  \end{equation}
  It remains to show that for every $\wsf\in \Wsf$
  \[
    e(X,Y|\wsf)\in\convexhull\circ\nbd_{\delta}\bigl(\extgr(\Gsf)\bigr)
  \]

  Fix an arbitrary $\wsf\in \Wsf$ and set $(X',Y'):=(X,Y|\wsf)$.  If
  $H(X'Y')<\log3$, then $\|e(X',Y')\|_{1}<3\log3$ and therefore
  $e(X',Y')\in\nbd_{3\log3}\extgr(\Gsf)$. So from now on we assume
  that $H(X'Y')\geq \log3$. In that case the support of $X'Y'$ has at
  least 3 points and we can apply Lemma~\ref{l:aud} to the pair
  $(X',Y')$ with $\epsilon = 1/\log |\Esf|$.  We obtain (see the
  remark after Lemma~\ref{l:aud}) an extension of $(X',Y',V)$ by $V$
  with alphabet
  \[
    \Vsf:=\set{v_{0},\dots,v_{k-1},v_{\infty}}
  \]
  such that \( H(V) = O(\log \log |\Esf|), \) and for every $i\in[k]$ the
  pair $(X',Y'|v_{i})$ is $\delta_{0}$-uniform and subsupported on
  $\Gsf$.
  
  First we examine the shift in entropy profile caused by conditioning on $V$ 
  \begin{equation}
    \label{eq:xy-xy|v}
    \begin{aligned}
      \|e(X',Y')-e(X',Y'|V)\|_{1}
      \leq
      3H(V)
      &\leq
        O( \log \log |\Esf|))\\
      &=
        O(\log H(XY)).
    \end{aligned}
  \end{equation}
  The last inequality follows from the assumption that $(X,Y)$ is
  uniform on its support $\Esf$. Next we split
  \begin{align*}
    e(X',Y'|V)=\sum_{i=0}^{k-1}p(v_{i})\cdot
    e(X',Y'|{v_{i}})+p(v_{\infty})\cdot e(X',Y'|{v_{\infty}})
  \end{align*}  
  Consider all terms of this sum except for the last one: 
  \[
    \sum_{i=0}^{k-1}p(v_{i})\cdot e(X',Y'|{v_{i}})
    =
    \sum_{i=0}^{k-1}p(v_{i})\cdot e(X',Y'|{v_{i}})
    +p(v_{\infty})\vect{0\\0\\0}
  \]
  The right hand side is the convex combination of points
  $e(X',Y'|{v_{i}})$, $i\in[k]$ and the point $(0,0,0)\in\Rbb^{3}$.
  Since by Lemma~\ref{l:unif-subsuported} for every $i\in[k]$
  \[
    e(X',Y'|v_{i})\in\nbd_{5\log\delta_{0}}\extgr(\Gsf)
  \]
  and $(0,0,0)\in\extgr(\Gsf)$, we conclude
  \[
    \sum_{i=0}^{k-1}p(v_{i})\cdot e(X',Y'|{v_{i}})
    \in
    \convexhull\circ\nbd_{5\log\delta_{0}}\extgr(\Gsf)
  \]
  On the other hand
  \[
    \|p(v_{\infty})\cdot e(X',Y'|v_{\infty})\|_{1}
    \leq
    3p(v_{\infty})\cdot H(X', Y' |v_{\infty})
    \leq
    3\epsilon \cdot \log|\Esf|
    \leq
    3
  \]
  Here we used the fact that $(X',Y')$, as well as $(X,Y)$, is distributed
  on $E$. Thus, we have
  \[
    e(X',Y'|V)
    \in
    \nbd_{3}\circ\convexhull\circ\nbd_{5\log\delta_{0}}\extgr(\Gsf)
  \]
  Since taking convex hull and tubular neighborhood commute in any
  normed vector space, we combine the last inclusion above
  with~\eqref{eq:xy-xy|v} to get
  \[
    e(X',Y')
    \in
    \convexhull\circ\nbd_{C'\log H(XY)+C''}\extgr(\Gsf)
  \]
  Together with inclusions~\eqref{eq:xy|w} we conclude that for any
  extension $(X,Y,W)$
  \[
    e(X,Y|W)
    \in
    \convexhull\circ\nbd_{C'\log H(XY)+C''}\extgr(\Gsf)
  \]
  Note that $\log H(XY) = \log \log |\Esf| \geq\log\log3>0$ and $C''$ may be absorbed by the
  $\log$-term.
  Thus
  \[
    \ext(X,Y)\subset\convexhull\circ\nbd_{C\log H(XY)}\extgr(\Gsf)
  \]
  \shiftup
\end{proof}

\section{Spectral bounds on the shape}\label{s:spectral}
\let\thetheoremSAVE=\thetheorem
\let\theequationSAVE=\theequation
\def\thetheorem{\thesection.\Alph{lemma}}
\def\theequation{\thesection.\Alph{lemma}}

As a culmination of the theory developed in the previous sections, in
this section we prove the spectral bound on the shape function of
random pair uniformly supported on the graph.  The proof of the
theorem in this chapter is relatively short, as the necessary
groundwork has been done in the preceding chapters.
\begin{theorem}\label{th:shape-spectral}
  There is a universal constant $C$ such that for any pair of random
  variables $(X,Y)$ uniformly supported on an admissible
  biregular bipartite graph $\Gsf=(\Xsf\sqcup\Ysf,\Esf)$ with the
  second largest eigenvalue $\lambda_{2}$ and any
  $\alpha,\beta\in[0,1]$, $\alpha+\beta\leq1$ the following bound
  holds
  \[
    \Scal(X,Y)(\alpha,\beta)\leq
    \max
    \left\{
      \begin{aligned}
        &\alpha\cdot H(X)+\beta\cdot H(Y)-I(X:Y),\\
        &\beta\cdot H(Y|X) + \alpha\log\frac{\lambda_2^{2}}{d_{1}},\\
        &\alpha\cdot H(X|Y) + \beta\log\frac{\lambda_2^{2}}{d_{2}}
      \end{aligned}
    \right\} + C\cdot \log H(XY),
  \]
  where $d_{i}$ are the degrees of $\Gsf$.
\end{theorem}
Note that the terms $\log\frac{\lambda_2^{2}}{d_{i}}$ measure the
deviation of the graph from being an expander. If the spectral terms
are ignored, then the the first summand in the right-hand
side is equal to $\Scal_{\min}(X,Y)$, see
Section~\ref{s:shape-lower-bound}.  
\begin{remark}
  The conclusion of the theorem above can be equivalently rewritten as
  \[
    \Scal(X,Y)(\alpha,\beta)\leq
    \max
    \left\{
      \begin{aligned}
        &\alpha\cdot H(X)+\beta\cdot H(Y)-I(X:Y),\\
        &(\beta-\alpha)H(Y|X) + 2\alpha\log\lambda_2,\\
        &(\alpha-\beta)H(X|Y) + 2\beta\log\lambda_2
      \end{aligned}
    \right\} + C\cdot \log H(XY).
  \]
  This form of inequality is especially convenient in the case $\alpha=\beta$.
\end{remark}
\begin{proof}
  To prove the theorem we will use Theorem~\ref{th:entropic-in-graph}, 
  which allows us to compare $\Scal(X,Y)$ with the support function of
  the extension profile of $\Gsf$. The latter is the union of its lower
  and the upper parts, and we use
  Propositions~\ref{p:graph-support-upper}
  and~\ref{p:graph-support-lower} to derive the necessary bounds. The
  detailed proof follows.

  Recall that we defined
  \[
    \Scal(X,Y)(\alpha,\beta):=\Lambda(\ext(X,Y))
  \]
  By Theorem~\ref{th:entropic-in-graph} there is a constant $C>0$ such that with
  the notation $\delta:=C\cdot \log H(XY)$ we have
  \[
    \ext(X,Y)\subset\convexhull\bigl(\nbd_{\delta}\extgr\Gsf\bigr) 
  \]
  For brevity, we adopt the following notation for the support
  functions of the extension profile of $\Gsf$ and its lower and upper
  parts:
  \[
    f:=\Lambda\bigl(\extgr(\Gsf)\bigr),
    \qquad
    \flo:=\Lambda\bigl(\extgrl(\Gsf)\bigr),
    \qquad
    \fup:=\Lambda\bigl(\extgru(\Gsf)\bigr)
  \]
  Then, using Proposition~\ref{p:legendre}\ref{i:lip} we obtain
  \begin{equation}\label{eq:L<f}
    \Scal(X,Y)(\alpha,\beta)\leq f(\alpha,\beta)+C\cdot\log H(XY).
  \end{equation}
  From Proposition~\ref{p:legendre}\ref{i:union} (max-additivity of
  the support function with respect to unions) we have 
  \begin{equation}\label{eq:f<max}
    \begin{aligned}
      f(\alpha,\beta)
      &\leq
        \max\set{\flo(\alpha,\beta),\fup(\alpha,\beta)}\\
      &\leq
        \max
      \left\{
        \begin{aligned}
          &\alpha\cdot[\Xsf]+\beta\cdot[\Ysf]-[\Xsf:\Ysf],\\
          &(\beta-\alpha)[\Ysf|\Xsf] + 2\alpha\log\lambda_2,\\
          &(\alpha-\beta)[\Xsf|\Ysf] + 2\beta\log\lambda_2
        \end{aligned}
            \right\} +\log2.
    \end{aligned}
  \end{equation}
  Here we used Propositions~\ref{p:graph-support-upper}
  and~\ref{p:graph-support-lower} to bound $\flo$ and $\fup$.

  Recall that in a biregular graph
  \[
    [\Xsf] = H(X),
    \quad
    [\Ysf] = H(Y)
    \quad
    \text{and}
    \quad
    H(Y|X) = \log d_1,
    \quad
    H(X|Y) = \log d_2.
  \]
  Therefore,~\eqref{eq:f<max} rewrites to 
  \begin{equation}\label{eq:reduced}
    f(\alpha,\beta)
    \leq
    \max
    \left\{
      \begin{aligned}
        &\alpha\cdot H(X)+\beta\cdot H(Y)-I(X:Y),\\
        &\beta H(Y|X) + \alpha\log\frac{\lambda_2^2}{d_1},\\
        &\alpha H(X|Y) + \beta\log\frac{\lambda_2^2}{d_2}
      \end{aligned}
    \right\} +\log2.
  \end{equation}
  A combination of~\eqref{eq:L<f} and~\eqref{eq:reduced} gives the required inequality.
\end{proof}

\let\thetheorem=\thetheoremSAVE
\let\theequation=\theequationSAVE

\section{Applications}\label{s:applications}
\subsection{MMRV inequality and entanglement}
\label{sec:mmrv}
Recall the MMRV inequality, \cite{makarychev2002new}, that holds for
any quintuple $(X,Y,A,B,W)$
\begin{equation*}
  \ing(X,Y,A,B)\geq-I(X:Y|W)-I(X:W|Y)-I(W:Y|X)
\end{equation*}
Define the \emph{entanglement}%
\footnote{In \cite{zhang3new} this quantity was denoted $W(X,Y)$, and
  its value was connected with \emph{approximate representation}
  of the mutual information between $X$ and $Y$.} %
of the pair $(X,Y)$ by
\begin{equation}\label{eq:entanglement}
  \mmrv(X,Y):=\inf_{W}\big(I(X:Y|W)+I(X:W|Y)+I(Y:W|X)\big)
\end{equation}
With this notation, the MMRV inequality reads
\begin{equation*}
  \ing(X,Y,A,B)\geq-\mmrv(X,Y)
\end{equation*}
When the entanglement is zero, that is
\[
  I(X:W|Y)=I(W:Y|X)=I(X:Y|W)=0,
\]
the MMRV inequality becomes the Ingleton inequality. This is coherent with
the fact that the pair satisfies the G\'acs-Körner extractable information
condition by the Double Markovity Lemma,
~\cite{csiszar2011information,makarychev2002new,kaced2013conditional}.
That is, there is $W$ such that
\[
  H(W|X)=H(W|Y)=I(X:Y|W)=0.
\]

\subsubsection{Spectral bound on the entanglement}
The entanglement of a random pair can be expressed in terms of its
shape function. 
\begin{proposition}\label{p:chinese-bound}
  For any pair $(X,Y)$ it holds
  \begin{equation*}
    \mmrv(X,Y)
    =
    2\Scal(X,Y)(\tfrac12,\tfrac12)-3 \Scal(X,Y)(\tfrac13,\tfrac13)
  \end{equation*}
\end{proposition}
\begin{proof}
  Recall that the value $\Scal(X,Y)(\tfrac12,\tfrac12)$ is
  uniquely determined by Shannon inequalities, and we
  have
  \[
    \Scal(X,Y)(\tfrac12,\tfrac12)=L(X,Y)=\frac12\big(H(X|Y)+H(Y|X)\big)
  \]
  see Proposition~\ref{p:L-upper}.
  On the other hand, for any extension $(X,Y,W)$, we have
  \begin{align*}
    I(X:W|Y)
    &=
    H(X|Y)-H(X|W)+I(X:Y|W)\\
    I(W:Y|X)
    &=
    H(Y|X)-H(Y|W)+I(X:Y|W)
  \end{align*}
  Therefore
  \begin{align*}
    &I(X:Y|W)+I(X:W|Y)+I(Y:W|X)\\
    &\qquad =
      H(X|Y)+H(Y|X) -\bigl(H(X|W)+H(Y|W)-3I(X:Y|W)\bigr)
  \end{align*}
  Taking the infimum over all extensions $W$ gives
  \[
    \mmrv(X,Y)
    =
    2L(X,Y)-3\Scal(X,Y)\left(\tfrac13,\tfrac13\right)
    =
    2\Scal(X,Y)\left(\tfrac12,\tfrac12\right)-3\Scal(X,Y)\left(\tfrac13,\tfrac13\right)
  \]
\end{proof}

Applying the spectral bound for $\Scal(X,Y)$, we get the following
spectral estimate for the entanglement, Theorem~\ref{th:main-intro} in
the introduction.
\begin{theorem}\label{th:entanglement-spectral}
  For any pair of random variables $(X,Y)$ uniformly supported on
  an admissible graph $\Gsf$ with the largest eigenvalue $\lambda_{1}$
  and the second largest eigenvalue
  $\lambda_{2}$ we have 
  \[
    \mmrv(X,Y)
    \geq
    \min\set{I(X:Y),2\log\frac{\lambda_{1}}{\lambda_{2}}}
    -O\bigr(\log H(XY)\bigl)
  \]
\end{theorem}
\begin{proof}
  Using the spectral bound on $\Scal(X,Y)$, Theorem~\ref{th:shape-spectral}, we obtain
  \[
    \Scal(X,Y)(\tfrac13,\tfrac13)
    \leq
    \max\set{\frac13 H(X)+\frac13
      H(Y)-I(X:Y),\frac23\log\lambda_{2}}+C\cdot\log H(XY)
  \]
  Then by Proposition~\ref{p:chinese-bound} 
  \begin{align*}
    \mmrv(X,Y)
    &=
      2\Scal(X,Y)\left(\tfrac12,\tfrac12\right)-3\Scal(X,Y)\left(\tfrac13,\tfrac13\right)\\
    &\geq
      H(X|Y)+H(Y|X)-\max\set{H(X)+H(Y)-3I(X:Y),2\log\lambda_{2}}\\
    &\quad
      -3C\cdot\log H(XY)\\
    &=
      \min\set{I(X:Y),2\log\frac{\lambda_{1}(\Gsf)}{\lambda_{2}(\Gsf)}}
      -O\bigr(\log H(XY)\bigl)
  \end{align*}
\end{proof}
\subsubsection{The case of an expander}
If the pair $(X,Y)$ is uniformly supported on a right-heavy expander $\Gsf$.
Then the conclusion of Theorem~\ref{th:entanglement-spectral} reads 
\begin{equation*}
  \mmrv(X,Y)
  \geq
  \min\set{I(X:Y),H(X|Y)}
  -O\bigl(\log H(XY)\bigr)
\end{equation*}
where implicit constants in $O\bigl(\log H(XY)\bigr)$ depend on
  the constant $C$ from the definition of an expander family in
  Section~\ref{s:expanders}. Importantly, the inequality
\begin{equation*}
  \ing(X,Y)\geq -\min\set{I(X:Y),H(X|Y)}
\end{equation*}
is Shannon-type by \cite[Lemma 8.A]{matveev2026spectral}. Thus, in the
case when $\Gsf$ is an expander, the MMRV inequality is rather weak,
i.e., it is equivalent to a Shannon-type inequality up to an additive
term $O(\log H(XY))$.

\subsection{When are  $\Scal$ and $\ext$ minimal or maximal possible?}
Recall that for a pair $(X,Y)$ we defined two sets
\[
  \ext_{\min}(X,Y),\ext_{\max}(X,Y)\subset\Rbb^{3}
\]
and two functions $\Scal_{\min}(X,Y)$, $\Scal_{\max}(X,Y)$ on the
square $[0,1]^{2}$, equal to the partial support functions of the above sets.  We
have the following bounds
\begin{align*}
  \ext_{\min}(X,Y)\subset\ext&(X,Y)\subset\ext_{\max}(X,Y)\\
  \Scal_{\min}(X,Y)\leq\Scal(&X,Y)\leq\Scal_{\max}(X,Y)
\end{align*}
where $\ext_{\min}(X,Y)$, $\ext_{\max}(X,Y)$, $\Scal_{\min}(X,Y)$
and $\Scal_{\max}(X,Y)$ depend only on the entropy profile of $X,Y$.

In this section we investigate under what conditions on the pair
the lower and upper bounds on the shape function are (almost) attained.

\subsubsection{When is the upper bound attained?}
We begin with a characterization of the pair with the maximal possible
extension profile.
\begin{theorem}\label{p:maxprofile}
  Let $(X,Y)$ be a pair of random variables. Then
  \[
    \Scal(X,Y)(\alpha,\beta)=\Scal_{\max}(X,Y)(\alpha,\beta)
  \]
  for all $(\alpha,\beta)\in[0,1]^{2}$ if and
  only if
  \[
    \Ecal(X,Y)=0
  \]
\end{theorem}
As we remarked before, pairs with zero entanglement satisfy
the G\'acs--Körner extractable mutual information condition.
Thus, Theorem~\ref{p:maxprofile} says that
$(X,Y)$ has extractable mutual information, if and only if
\[
  \Scal(X,Y) =\Scal_{\max}(X,Y)
\]
\begin{proof}
  On the upper-right triangle,
  $\set{0\leq\alpha,\beta\leq1,\;\alpha+\beta\geq1}$, the function
  $\Scal(X,Y)$ is affine and determined by the entropy profile of
  the pair.
  Therefore we have to examine its behavior only on the lower-left
  triangle.
  
  Suppose that
  \[
    \Scal(X,Y)(\alpha,\beta)
    =
    \alpha\cdot H(X|Y)+\beta\cdot H(Y|X)
  \]
  for all $(\alpha,\beta)\in[0,1]^2$ with $\alpha+\beta\leq1$. Choose
  an interior point of this triangle, say $(\alpha_{0},\beta_{0})=(1/3,1/3)$.
  Since $\ext(X,Y)$ is compact,  there
  exists an extension $(X,Y,W)$ such that
  \[
    \alpha_{0}\cdot H(X|W)+\beta_{0}\cdot H(Y|W)-I(X:Y|W)
    =
    \alpha_{0}\cdot H(X|Y)+\beta_{0}\cdot H(Y|X).
  \] 
  where the right-hand side is exactly the value of
  $\Scal_{\max}(\Xbf)$ at $(\alpha_{0},\beta_{0})$.
  The left-hand side can be rewritten as
  \[
    \alpha_{0}\cdot H(X|YW)+\beta_{0}\cdot H(Y|XW)
    +(\alpha_{0}+\beta_{0}-1)I(X:Y|W)
  \]
  Thus we have
  \begin{equation}
    \begin{aligned}
      \alpha_{0}\cdot I(X:W|Y)+\beta_{0}\cdot I(Y:W|X)+(1-\alpha_{0}-\beta_{0})I(X:Y|W)=0
    \end{aligned}
  \end{equation}
  Since $\alpha_{0}>0$, $\beta_{0}>0$, and $1-\alpha_{0}-\beta_{0}>0$,
  the equality forces 
  \[
    I(X:W|Y)=I(Y:W|X)=I(X:Y|W)=0.
  \]
  This implies $\mmrv(X,Y)=0$.
  
  Conversely, if $\mmrv(X,Y)=0$, then there exists an extending variable
  $W$ satisfying the G\'acs--Körner condition, see p.~\pageref{gacs-korner}. Then
  \begin{align*}
    \Scal(X,Y)(\alpha,\beta)
    &\geq
      l_{\alpha\beta}\big(H(X|W),H(Y|W),I(X:Y|W)\big)\\
    &=
      \alpha\cdot H(X|Y)+\beta\cdot H(Y|X)
  \end{align*}
\end{proof}
\begin{corollary}
  The essential part of the extension profile of a pair $(X,Y)$ is
  maximal if and only if $(X,Y)$ satisfies the G\'acs--Körner
  extractable information condition.
\end{corollary}

\subsubsection{When is the shape function minimal?}
We proceed with a discussion of pairs with (approximately) minimal
possible extension profile. Recall that $\Gsf$ belongs to an expander
family if there is constant $C\geq0$ such that one of the conditions
is satisfied for all graphs in the family
\begin{align}
  &2\log\lambda_{2}(\Gsf)\leq \log d_{1}(\Gsf)+C
    \tag{right-heavy expander}\\
  &2\log\lambda_{2}(\Gsf)\leq \log d_{2}(\Gsf)+C
    \tag{left-heavy expander}\\
  &2\log\lambda_{2}(\Gsf)\leq \log \lambda_{1}(\Gsf)+C
    \tag{balanced expander}
\end{align}
\begin{theorem}\label{th:L-min}
  Let $(X,Y)$ be uniformly supported on a right-heavy expander
  $\Gsf$. Then for every $\alpha,\beta\in[0,1]$
  \[
    \Scal(X,Y)(\alpha,\beta)
    \leq
    \max
    \left\{
      \begin{aligned}
        &\alpha\cdot H(X)+\beta\cdot H(Y)- I(X:Y),\\
        &\beta\cdot H(Y|X),\\
        &\alpha\cdot H(X|Y)+\beta\bigl(H(Y|X)-H(X|Y)\bigr)
      \end{aligned}
    \right\}
    +O\bigl(\log H(XY)\bigr)
  \]
\end{theorem}
Note that the two first terms in the maximum coincide with those in
the definition of $\Scal_{\min}(P)$. Thus, for pairs supported on
expanders, $\Scal(X,Y)$ partially coincides with
$\Scal_{\min}(X,Y)$ up to an error logarithmic in $H(XY)$.  We would
like to conjecture that, in fact, $\Scal(X,Y)$ is equal to
$\Scal_{\min}(X,Y)$ on all of the square up to an error logarithmic in $H(XY)$.
The conjecture holds true, if we assume in addition that expanders are
balanced.
\begin{corollary}\label{cor:L<Lmin}
  If $(X,Y)$ is uniformly supported on a balanced expander, then
  for every $\alpha,\beta\in[0,1]$
  \[
    \Scal(X,Y)(\alpha,\beta)
    \leq
    \Scal_{\min}(X,Y)(\alpha,\beta)+O(\log H(XY))
  \]
\end{corollary}
\begin{proof}[Proof of Theorem~\ref{th:L-min}]
  If $\alpha+\beta\le 1$, we apply Theorem~\ref{th:shape-spectral}
  (see also the Remark after this theorem) and obtain the bound
  \[
    \Scal(\Xbf)(\alpha,\beta)\leq
    \max
    \left\{
      \begin{aligned}
        &\alpha\cdot H(X)+\beta\cdot H(Y)-I(X:Y),\\
        &(\beta-\alpha)H(Y|X) + 2\alpha\log\lambda_2,\\
        &(\alpha-\beta)H(X|Y) + 2\beta\log\lambda_2
      \end{aligned}
    \right\} + O\bigl(\log H(XY)\bigr)
  \]
  Substituting the right-heavy expander condition
  \[
    2\log\lambda_{2}\leq \log d_{1}+C=H(Y|X)+C
  \]
  we obtain the conclusion of the theorem.
  
  If $\alpha+\beta\ge 1$, we apply Proposition~\ref{p:L-upper}~\ref{i:ur-triangle}.
\end{proof}
Corollary~\ref{cor:L<Lmin} is a special case of Theorem~\ref{th:L-min}.
Recall that we call $(X,Y)$ \emph{rigid} if
\[
  \Scal(X,Y)=\Scal_{\min}(X,Y)
\]
In effect we have shown, that pairs supported on the balanced expanders are
rigid up to a logarithmic error.

\subsection{Zhen Zhang's question on approximate
  representation of the mutual information}
\label{s:applications-zhang}

As discussed above, for a pair of random variables $(X,Y)$, the
quantity $\mmrv(X,Y)$ defined in \eqref{eq:epsilon(x,y)} measures the
extent to which the mutual information between $X$ and $Y$ can be
materialized.  At one extreme, when $\mmrv(X,Y)=0$, the mutual
information is fully materializable (see p.~\pageref{sec:mmrv}).

We now turn to the opposite extreme, in which $\mmrv(X,Y)$ is close to
its maximal value.  In this case, one can say that the
mutual information is maximally non-extractable.

The entanglement satisfies the upper bound
\[
  \mmrv(X,Y)\leq\min\set{H(X|Y),H(Y|X),I(X:Y)}, 
\]
which is obtained by restricting the minimization in
\eqref{eq:entanglement} to $W$ in the list $\set{X,Y, \bullet}$.

In \cite{zhang3new}, Zhen Zhang raised the following natural question:
under what conditions on the quantities $ H(X), H(Y), I(X:Y) $ can the
exact equality $\mmrv(X,Y)=\min\set{H(X|Y),H(Y|X),I(X:Y)}$ hold?%
\footnote{In \cite{zhang3new} a more restricted version of the
  question is raised, namely under what conditions the equality
  $\mmrv(X,Y)=I(X:Y)$ may hold. We address here a slightly more
  general question of which Zhang's question is a part.}  %
More specifically, which ratios (homogeneous coordinates)
\[
  \left[ H(X) : H(Y) : I(X:Y) \right]
\]
can arise from a pair of random variables $(X,Y)$ satisfying
\[
  \mmrv(X,Y)=\min\set{H(X|Y),H(Y|X),I(X:Y)}?
\]
Let us consider a relaxed version of Zhang's question.  We say that a
triple of coordinates $[x:y:z]$ (in the projective space
$\mathbb{RP}^2$) is $\delta$-\emph{realizable}, if there exists a pair
of random variables $(X,Y)$ such that
\[
\big[ H(X) : H(Y) : I(X:Y) \big]  = [x:y:z],
\]
and
\[
 \mmrv(X,Y)  \ge \min\set{H(X|Y),H(Y|X),I(X:Y)}-\delta\cdot H(XY)
\]
We show below that $\delta$-realizability essentially puts no
  restriction on projectivized entropy profiles of the pairs $(X,Y)$.

The following statement follows from \cite[Theorem
5.6]{marcus2013interlacing} and the accompanying construction in its
proof.
\begin{theorem}
\label{th:expanders}
  For all integers $k,l\geq 3$ and $n\geq1$ 
  there exists a biregular bipartite
  graph $\Gsf=(\Xsf\sqcup\Ysf,\Esf)$ such that
  \begin{align*}
    &d_{1}(\Gsf)=k,
      &&d_{2}(\Gsf)=l,\\
    &|\Xsf|=l\cdot 2^{n},
      &&|\Ysf|=k\cdot 2^{n},\\
    &\frac{|X|\cdot|Y|}{|E|}=2^{n},
    &&2\log\lambda_{2}(\Gsf)\leq\log\max\set{k,l}+\log 4
  \end{align*}
\end{theorem}

\begin{remark}
  \begin{enumerate}[label=(\roman*)]
  \item Note that the graphs provided by the theorem above are
    admissible and form an expander family with a uniformly bounded gap
    \[
      2\log\lambda_{2}-\log\max\set{d_{1},d_{2}}\leq\log 4.
    \]
  \item For any $C\geq3$ the set of projectivized $\log$-sizes of
    these graphs
    \[
      \Big\{\big[(\log l+n\log2):(\log k+n\log2):(n\log2)\big]
      \st
      k,l\geq C,\;n\geq1\Big\}
    \]
    is dense in the region
    \[
      \Big\{[x:y:z]\st 0\leq z\leq\min\set{x,y}\Big\}\subset\Rbb P^{2}
    \]
    cut out by the trivial inequalities on the cardinalities
    of $\Xsf, \Ysf,\Esf$ in the graph.
  \end{enumerate}
\end{remark}
Theorem~\ref{th:entanglement-spectral} implies that there is a
universal constant $C>0$ such that for $(X,Y)$
uniformly distributed on $\Gsf$ provided by Theorem~\ref{th:expanders}
it holds
\[
  \mmrv(X,Y)\geq \min\set{H(X|Y),H(Y|X),I(X:Y)}- C\cdot\log H(XY).
\]
Thus, the projective triple of coordinates
\[
  \big[H(X) :  H(Y) :  I(X:Y)\big]
\]
is $\delta$-realizable for
\[
\delta=C\cdot\frac{\log H(XY)}{H(XY)}
\]

Therefore we have proved the following theorem.
\begin{theorem}\label{th:zhang-region}
  For every $\delta>0$, the $\delta$-realizable
  triples
  \[
    \big[H(X):H(Y):I(X:Y)\big]
  \]
  are dense in the region determined
  by the Shannon inequalities
  \[
    \Big\{[x:y:z]\st 0\leq z\leq\min\set{x,y}\Big\}\subset\Rbb P^{2}
  \]
\end{theorem}

The approaches based on Theorems~\ref{th:entanglement-spectral} and~\ref{th:expanders}
does not allow us to reduce the $\delta$-term to zero.
Thus, the Zhang's original question on precise realizability with $\delta=0$ remains open.


\appendix

\section{Proof of Lemma~\ref{l:aud}}
In this section we will show how to decompose a pair of random
variables in a mixture of almost-uniform ones and a ``negligible remainder.''
The proof follows rather closely the argument of a similar statement from \cite{matveev2026spectral}.
However, the proof proposed in this section and the argument published in  \cite{matveev2026spectral}
are formally independent (one does not follow from another one).

\subsection{Preliminaries~1: regular sets}
We begin with a few definitions, following the terminology from \cite{matveev2026spectral}.

\begin{definition}
For a set $\Ssf\subset \Xsf \times \Ysf$  and a fixed $x_0\in\Xsf$
we denote by $\Ssf|_{x_0} $ the set of all pairs $(x_0,y)$ in $\Ssf$, i.e., 
\(
\Ssf|_{x_0} : =  \Ssf\cap (\{x_0\} \times \Ysf).
\)
Similarly, for $y_0\in\Ysf$ we denote
\(
\Ssf|_{y_0} : =  ( \Xsf\times \{y_0\} )\cap \Ssf. 
\)

Further, we say that $\Ssf$
is \emph{$\delta$-regular} for some $\delta\geq1$ if 
\[
  \frac{\max_{x\in \Xsf} \sharp(\Ssf|_x)}
  {\min_{x\in \Xsf} \sharp(\Ssf|_x)}
  \leq \delta
  \quad \text{and} \quad 
    \frac{\max_{y\in \Ysf} \sharp(\Ssf|_y)}
  {\min_{y\in \Ysf} \sharp(\Ssf|_y)}
  \leq \delta, 
\]
where $\sharp(\cdot)$ stands for cardinality of the set and
  extrema are taken over parameter ($x$ or $y$) ranging over the
  projection of $S$ to the corresponding coordinate set.
\end{definition}

The following theorem is the main technical tool of this section.
\begin{theorem}[{the two-dimensional case of \cite[Theorem
    3]{alon2007partitioning}}]
  \label{p:cool}
  There exist constants $\alpha\geq 1$ and
  $\beta\geq1$ such that every $2$-dimensional finite set $\Ssf$
  containing at least three points can
  be partitioned into at most $(\log|\Ssf|)^{\alpha}$ parts so that
  each part is $\beta$-regular.
\end{theorem}

\begin{remark}
  The original theorem in \cite{alon2007partitioning} applies to
  $\Ssf$ of any size but claims that the number of parts is
  $\mathrm{poly}\log|\Ssf|)$. We require $S$ to contain at least
  three points. This condition allows us to replace a polynomial by a
  simpler expression $(\log|\Ssf|)^{\alpha}$, avoiding additive and
  multiplicative constants.
\end{remark}

\subsection{Preliminaries~2: regular distributions}

\begin{definition}
For a real number $\sigma>1$,
a random variables $A$ distributed on $\Asf$ is called $\sigma$-\emph{uniform} if  
\[
\frac{\max_a \prob[A=a]}{\min_a \prob[A=a]} \le \sigma,
\]
where the extrema are taken over the set of all atoms $a\in \Asf$.
\end{definition}
Observe that the definition of $\sigma$-uniformity applies verbatim to multi-dimensional probability distributions 
(e.g., to jointly distributed pairs of random variables).

\begin{definition}
A pair of jointly distributed random variables $(X,Y)$ distributed on $\Xsf\times \Ysf$ is called $\delta$-\emph{regular} if  
the support of $(X,Y)$ is a $\delta$-regular set.
\end{definition}
The next lemma follows easily form the definitions. 
\begin{lemma} (see \cite[Lemma~1(ii)]{matveev2026spectral})
\label{l:regular-to-uniform}
 If $(X,Y)$ is $\delta$-regular and $\sigma$-uniform, 
then both marginal distributions $X$ and $Y$ are $(\delta\sigma)$-uniform. 
\end{lemma}

\subsection{Splitting a single random variable}\label{s:split1}

Next proposition is an adjustment of
\cite[Proposition~4.1.A]{matveev2026spectral} to our settings.  The
changes in the proof are minimal.

\begin{proposition}\label{p:split1}
  For every random variable $X$ distributed on $\Xsf$ and for every
  $\epsilon>0$ there exists an extension $(X,U)$ such that the
  alphabet of $U$ is $\Usf=\set{u_{0},u_{1},\dots,u_{k-1},u_{\infty}}$
  and the following properties hold:
  \begin{enumerate}[label=(\roman*)]
  \item \label{i:partition}
    $H(U|X)=0$; 
  \item \label{i:smalltail}
    $\displaystyle\Pbb[U=u_{\infty}] \leq \epsilon$; 
  \item \label{i:entropybound}
    $k\leq \frac{H(X)}\epsilon +1$ ;
  \item \label{i:uniform}
    $X|u_i$ is 2-uniform for every $i\in[k]$; 
  \item \label{i:smallparts}
    $\sharp(X|u_{i})\leq 2^{ \frac{H(X)}{\epsilon}+1 } $ for every $i\in[k]$. 
  \end{enumerate}
  \end{proposition} 
  
\begin{proof}
  Let $p$ be the distribution mass function of $X$.
  We construct $U$ by partitioning the alphabet
  $\Xsf$ of $X$, where the alphabet of $U$ is  the set of
  parts of the partitioning;  $U$ will be deterministic function of $X$,
  so property~\ref{i:partition} follows immediately. We will
  partition $\Xsf$ as
  \begin{equation*}
    \Xsf=u_{0}\sqcup u_{1}\sqcup\dots\sqcup u_{k-1}\sqcup u_{\infty}
  \end{equation*}
  where
  \begin{equation*}
    \begin{aligned}
      u_{i}
      &:=
        \set{x\in\Xsf\st 2^{-(i+1)}<p(x)\leq 2^{-i}},
      &&
         i=0,\dots,k-1\\
      u_{\infty}
      &:=
        \set{x\in\Xsf\st p(x)\leq 2^{-k}}, 
    \end{aligned}
  \end{equation*}
  and the value of $k$ is chosen later in the proof.
  If some of the parts defined above are empty, we silently drop  them.
  We set
  \[
    q_{i}:=\Pbb[U=u_{i}]=\sum_{x\in u_{i}} p(x)>0,
    \quad
    i=0,\ldots,k-1,\infty.
  \]
  Note that for $x,y\in u_{i}$, $i\in[k]$, the probabilities $p(x)$
  and $p(y)$ differ at most 2-fold, so the conditional distribution $X|u_{i}$ is 2-uniform
  and property~\ref{i:uniform} is satisfied.

  It remains to choose the value of $k$ so that
  properties~\ref{i:smalltail},~\ref{i:entropybound}
  and~\ref{i:smallparts} of the proposition hold.  To find an
  appropriate value of $k$, we write
  \begin{equation*}
     H(X)
    \geq
    \sum_{x\in u_{\infty}}p(x)\log p(x)^{-1}
    \geq
    q_{\infty}\cdot\log2^{k}
    =k\cdot q_{\infty}\cdot\log2. 
  \end{equation*}
  We take $k:=\left\lceil \frac{H(X)}{\epsilon \log 2}\right\rceil$.  Observe that
   property~\ref{i:entropybound} follows immediately.  Further, for the chosen $k$
   \[
    q_{\infty} \leq\frac{ H(X)}{k \cdot \log2} \le \epsilon, 
  \]
  and property~\ref{i:smalltail} is satisfied.
  We also observe that
  $|u_{i}|\leq q_{i}\cdot2^{i+1}$. 
  Thus,
  \begin{equation*}
    \sharp(X|u_{i})\leq q_{i}\cdot 2^{i+1}\leq 2^{k} \leq2^{ \frac{H(X)}{\epsilon}+1 } , 
  \end{equation*}
  and property~\ref{i:smallparts} and proposition as a whole are
  proven.
\end{proof}

\subsection{Slicing a distribution into almost-uniform components}\label{s:split4}

Now we are ready to prove the Almost-Uniform Decomposition Lemma.

\begin{proof}[Proof of Lemma~\ref{l:aud}]
  We begin with considering $XY$ as one random variable, ignoring its $2$-dimensional structure. 
  We apply Proposition~\ref{p:split1}  and  obtain a partition 
  \[
  \Usf=\set{u_{0},u_{1},\dots,u_{k'-1},u_{\infty}}
  \]
  of the set $\Esf$ that satisfies properties
  \ref{i:partition}--\ref{i:smallparts}.  In particular, by
  Proposition~\ref{p:split1}\ref{i:entropybound}, the cardinality $k'$
  of this partition is bounded by 
  \(
  \frac{ H(X,Y)}\epsilon +1 .
  \)  
  (Here $k'$ is used to
  distinguish this quantity from the parameter $k$ appearing in the
  statement of the lemma.)

  We now refine this partition
  using Theorem~\ref{p:cool}.
  We leave $u_{\infty}$ unchanged and set
  $v_{\infty}:=u_{\infty}$, thus guaranteeing
  property~\ref{i:smalltail4} in the conclusion of the lemma.

  Consider one of the parts $u_{i}$, $i\in[k']$.  By
  Proposition~\ref{p:split1}\ref{i:uniform} and \ref{i:smallparts},
  $XY|u_{i}$ is 2-uniform and its support contains at most
  $2^{ \frac{ H(X,Y)}\epsilon +1 }$ points. If $|u_{i}|<3$ then
    the support of $(XY|u_{i})$ is already 1-regular, as every
    multidimensional set consisting of one or two points is
    1-regular. In that case we don't further partition $u_{i}$ and set
    $v_{i}:=u_{i}$. Otherwise we apply Theorem~\ref{p:cool} to
  \[
    \Ssf :=u_{j}\subset\ \Esf .
  \]
  The theorem provides the partition of $u_{i}$ into at most
  \begin{equation}
  \label{eq:upper-bound-parts}
    \left(\log2^{  \frac{ H(X,Y)}\epsilon +1 }\right)^{\alpha}
    =
   \left(   \frac{ H(X,Y)}\epsilon +1 \right)^{\alpha}
  \end{equation}
  $\beta$-regular parts $v_{i1},v_{i2},\dots$ 
  Thus, $(X,Y|v_{ij})$
  is a $\beta$-regular pair for every $j$ in its range. 
  
  Note also that $(X,Y)|v_{ij}$ is 2-uniform. 
  Therefore, by  Lemma~\ref{l:regular-to-uniform}, the pair
  $(X,Y)$ conditional on $v_{ij}$ is $2\beta$-uniform. 
  Thus, we can take
  $\delta_0:=2\beta$, proving property~\ref{i:regular4}.
  Refining all parts $u_i$, $i\in[k']$ gives at most
  \[
    D:= \left(  \frac{ H(X,Y) }\epsilon +1 \right)^{\alpha+1}
  \]
  parts (the number $D$ is the product of the upper bound on
    $k'$ and the upper bound on the number of almost-regular parts
    given in \eqref{eq:upper-bound-parts}). Therefore,
  \[
    H(V)
    \leq
    \log (D+1)
    \leq
    C_0\cdot\log \left(  \frac{H(X,Y) }\epsilon  \right) 
  \]
  for some $C_0>0$ (here we use the assumption
  $H(X,Y)\geq\log 3 > 1 $, which allows do absorb the additive
  constants in the expression by increasing the factor of $C_0$).
This proves the required inequality in~\ref{i:entropybound4}.
\end{proof}

\printbibliography[heading=bibintoc]

\end{document}